\documentclass[final,onecolumn,12pt]{IEEEtran}
\usepackage[colorlinks, linkcolor=red, citecolor=blue]{hyperref}
\usepackage{amssymb,amsthm,amsmath,latexsym,pslatex,cite}
\usepackage{lineno,hyperref,mathtools}
\usepackage{pdflscape}
\usepackage{microtype}
\usepackage{setspace}
\usepackage{amsmath}
\usepackage{amssymb}
\usepackage{bm}
\usepackage{amssymb}
\usepackage{amsthm}
\usepackage{multirow,booktabs}
\usepackage{cases}
\usepackage{tabularx}
\usepackage{adjustbox}
\usepackage[figuresright]{rotating}
\usepackage{longtable}
\usepackage{threeparttable}
\usepackage{booktabs}
\usepackage{multirow}
\usepackage{color}
\usepackage{cite}
\usepackage{multirow,booktabs}
\usepackage{cases}
\usepackage{tabularx}
\usepackage{adjustbox}
\usepackage[figuresright]{rotating}
\usepackage{longtable}
\usepackage{booktabs}
\usepackage{multirow}
\usepackage{color}
\usepackage{lineno,hyperref,mathtools}
\usepackage{soul}
\usepackage[ruled,linesnumbered]{algorithm2e}
\usepackage{algpseudocode}
\usepackage{pifont}
\usepackage{booktabs}

\allowdisplaybreaks

\newtheorem{theorem}{Theorem}

\newtheorem{question}{Question}
\newtheorem{lemma}[theorem]{Lemma}
\newtheorem{corollary}[theorem]{Corollary}

\theoremstyle{definition} 
\newtheorem{remark}{Remark}
\newtheorem{example}{Example}
\newtheorem{definition}{Definition}

\theoremstyle{remark}

\newcommand{\F}{\mathbb{F}}

\newcommand{\C}{{\mathcal{C}}}

\newcommand{\SSS}{{\mathcal{S}}}

\newcommand{\GRS}{{\mathrm{GRS}}}
\newcommand{\EGRS}{{\mathrm{EGRS}}}

\newcommand{\RL}{{\mathrm{RL}}}

\newcommand{\vvv}{{{\mathbf{v}}}}

\newcommand{\wt}{{{\rm{wt}}}}

\newcommand{\ignore}[1]{}

\makeatletter
\newcommand{\rmnum}[1]{\romannumeral #1}
\newcommand{\Rmnum}[1]{\expandafter\@slowromancap\romannumeral #1@}
\makeatother

\begin{document}

\title{A framework for constructing non-GRS MDS-NMDS codes from deep holes and its application}
\author{Yang Li, Zhenliang Lu, San Ling, Kwok-Yan Lam
\thanks{Yang Li is with the School of Physical and Mathematical Sciences, 
Nanyang Technological University, 21 Nanyang Link, Singapore 637371, Singapore
(email: yanglimath@163.com).} 
\thanks{Zhenliang Lu is with the Department of Computer Science, City University of Hong Kong, 
Hong Kong, China (email: zhenliang.lu@cityu.edu.hk).}
\thanks{San Ling is with the School of Physical and Mathematical Sciences, 
Nanyang Technological University, 21 Nanyang Link, Singapore 637371, Singapore
(email: lingsan@ntu.edu.sg). 
He is also with VinUniversity, Vinhomes Ocean Park, Gia Lam, Hanoi 100000, Vietnam (email: ling.s@vinuni.edu.vn).} 
\thanks{Kwok-Yan Lam is with the Digital Trust Centre,
Nanyang Technological University, 50 Nanyang Drive, Singapore 639798, Singapore 
(email: kwokyan.lam@ntu.edu.sg).}
\thanks{
    This research is supported by the Nanyang Technological University Research Grant under No. 04INS000047C230GRT01  
    and the National Research Foundation, Singapore and Infocomm Media Development Authority 
    under its Trust Tech Funding Initiative. 
    Any opinions, findings and conclusions or recommendations expressed in this material are those of the authors 
    and do not reflect the views of National Research Foundation, Singapore and Infocomm Media Development Authority. 
    }
}



\maketitle
\begin{abstract}
    Maximum distance separable (MDS) codes and near MDS (NMDS) codes are of particular interest in coding theory 
    due to their optimal error-correcting capabilities and wide applications in communication, cryptography, 
    and storage systems. A family of linear codes is called a family of non-GRS MDS-NMDS codes 
    if for each $[n,k]_q$ code in the family, it is either an $[n,k,n-k+1]_q$ MDS code that is not 
    monomially equivalent to any GRS code or extended GRS code, or an $[n,k,n-k]_q$ NMDS code. 
    This paper develops a unified framework for constructing new families of non-GRS MDS-NMDS codes 
    via deep holes. We show that, starting from a family of $[n,k]_q$ non-GRS MDS-NMDS codes with covering radius $n-k$, 
    one can systematically obtain more $[n+1,k+1]_q$ non-GRS MDS-NMDS codes. 
    The proposed framework is further reformulated in terms of the second kind of extended codes. 
    {This reformulation recovers a main result of Wu, Ding, and Chen (IEEE Trans. Inf. Theory, 71(1): 263–272, 2025), 
    provides a provable reduction in the computational complexity compared with the approach of Ma, Kai, and Zhu 
    (Finite Fields Appl., 114, 102844, 2026), 
    and reveals additional structural properties of the resulting codes.} 
    As an application, we determine the covering radius and characterize two classes of deep holes of 
    extended subcodes of GRS codes. 
    By applying our framework, we obtain three new families of non-GRS MDS–NMDS codes 
    and investigate the monomial equivalence between the resulting codes and Roth–Lempel codes.
\end{abstract}

\begin{IEEEkeywords}
    Non-GRS MDS code, NMDS code, covering radius, deep hole, second kind of extended code
\end{IEEEkeywords}



\section{Introduction}

Throughout this paper, let $\F_q$ be the finite field of size $q$ 
and let $\F_q^n$ be the $n$-dimensional vector space over $\F_q$, 
where $q=p^m$ is a prime power and $n$ is a positive integer.  
An $[n,k,d]_q$ {\em linear code} $\C$ is a $k$-dimensional linear subspace of $\F_q^n$ 
with minimum distance $d=d(\C)$, and it satisfies the {\em Singleton bound}: $d\leq n-k+1$ \cite{HP2003}.    
The code $\C$ is called a {\em maximum distance separable (MDS)} code if $d = n - k + 1$, 
and is called a {\em near MDS (NMDS)} code if $d(\C)=n-k$ and $d(\C^{\perp})=k$, 
where $\C^{\perp}$ is the dual of $\C$ defined by
\[
\C^{\perp}
=\left\{{\bf x}=(x_1,x_2,\ldots,x_n)\in \F_q^n:~ 
\sum_{i=1}^{n} x_i y_i = 0,~ \forall~{\bf y}=(y_1,y_2,\ldots,y_n)\in \C\right\}.
\]
It is also well known that $\C^{\perp}$ is an $[n,n-k]_q$ MDS (resp. NMDS) code 
if and only if $\C$ is an $[n,k]_q$ MDS (resp. NMDS) code.

MDS and NMDS codes are two special types of codes in coding theory, 
which have garnered significant attention due to their optimal error-correcting capabilities, 
elegant algebraic structures and wide-ranging applications in distributed storage systems \cite{CHL2011}, 
random error channels \cite{ST2013}, informed source and index coding problems \cite{TR2018}, 
and secret sharing schemes \cite{SV2018,ZWXLQY2009}.
With these facts, it is of particular interest to construct and study families of linear codes 
in which each code is either MDS or NMDS. 
For brevity, we refer to this family of codes as {\em MDS-NMDS codes} with the specific definition as follows. 

\begin{definition}\label{def.MDS-NMDS}
    Given a family of $[n,k]_q$ linear codes $\C_k$ with $1\leq k\leq n$, 
    if both the sets 
    $$\{1\leq k\leq n:~\C_k~{\rm is~an~} [n,k,n-k+1]_q~ {\rm MDS~code}\}~{\rm and}~
    \{1\leq k\leq n:~\C_k~{\rm is~an~} [n,k,n-k]_q~ {\rm NMDS~code}\}$$ are non-empty,
    then we call $\C_k$ a family of {\em MDS-NMDS codes}. 
\end{definition}

\subsection{Non-GRS MDS-NMDS codes}

Let $\SSS=\{a_1,a_2,\ldots,a_n\}\subseteq \F_q$ be a set of $n$ distinct elements 
and let $\mathbf{v}=(v_1,v_2,\ldots,v_n)\in (\F_q^*)^n$ be a generic representation of its element. 
Let $\F_{q}[x]_k=\{f(x)=\sum_{i=0}^{k-1}f_ix^i:~f_i\in \F_q,~\forall~ i=0,1,\ldots,k-1\}.$
For any two positive integers $k$ and $n$ satisfying $1\leq k\leq n\leq q$,   
there is an $[n,k,n-k+1]_{q}$ {\em generalized Reed-Solomon (GRS) code} $\GRS_k(\SSS,\mathbf{v})$ 
such that 
$
\GRS_k(\SSS,\mathbf{v}) =  \{(v_1f(a_1),v_2f(a_2),\ldots,v_nf(a_n)):~ f(x)\in {\mathbb{F}_{q}[x]_{k}}\}.
$
By adding a coordinate to each codeword of a GRS code, 
we immediately get an $[n+1,k,n-k+2]_q$ MDS code, 
namely, an {\em extended GRS (EGRS) code} 
$
\GRS_k(\SSS,\mathbf{v},\infty) =  \{(v_1f(a_1),v_2f(a_2),\ldots,v_nf(a_n),f_{k-1}):~ f(x)\in {\mathbb{F}_{q}[x]_{k}}\},
$
where $f_{k-1}$ is the coefficient of $x^{k-1}$ in $f(x)$.     
From \cite{HP2003}, we write generator matrices of $\GRS_k(\SSS,\mathbf{v})$ 
and $\GRS_k(\SSS,\mathbf{v},\infty)$ as 
\begin{align}\label{eq.GRS.generator matrix}
    G_{\GRS_k(\SSS,\mathbf{v})}=\begin{pmatrix}    
        v_1 & v_2 & \ldots & v_{n} \\ 
        v_1a_1 & v_2a_2 &  \ldots & v_{n}a_{n} \\
        \vdots & \vdots &  \ddots & \vdots  \\
        v_1a_1^{k-1} & v_2a_2^{k-1} &  \ldots & v_{n}a_{n}^{k-1}  \\
    \end{pmatrix}~{\rm and}~
    G_{\GRS_k(\SSS,\mathbf{v},\infty)}=\begin{pmatrix}    
        v_1 & v_2 &   \ldots & v_{n} & 0 \\ 
        v_1a_1 & v_2a_2 &   \ldots & v_{n}a_{n} & 0 \\
        \vdots & \vdots &  \ddots & \vdots & \vdots \\
        v_1a_1^{k-2} & v_2a_2^{k-2} &  \ldots & v_{n}a_{n}^{k-2} & 0 \\
        v_1a_1^{k-1} & v_2a_2^{k-1} &  \ldots & v_{n}a_{n}^{k-1} & 1 \\
    \end{pmatrix},
\end{align}
respectively. 
GRS and EGRS codes are two special classes of MDS codes that 
have garnered a lot of attention due to their highly efficient encoding and decoding algorithms. 
Moreover, $[n,k,n-k+1]_q$ GRS and EGRS codes satisfy $n \leq q+1$, and together with the MDS conjecture $n \leq q+2$ 
in general except for two special cases \cite{HP2003}, they realize nearly all possible parameter sets of MDS codes, 
except for the case $n = q+2$.

However, it is still very interesting and vital to construct 
MDS codes that are not monomially equivalent to GRS and EGRS codes, referred to as {\em non-GRS MDS codes}. 
These non-GRS MDS codes are important in classifying MDS codes \cite{HP2003} 
and constructing quantum stabilizer codes with larger distances \cite{G2022,LCEL2022}. 
Additionally, although some non-GRS MDS codes with particular structures have been 
shown to be vulnerable in McEliece-type cryptanalytic schemes \cite{LR2020,CPTZ2025}, 
some efforts have confirmed that non-GRS MDS codes can exhibit significant potential 
in resisting cryptographic attacks such as the Sidelnikov–Shestakov and Wieschebrink attacks, unlike GRS codes \cite{BBPR2018,ZJ2025}.
{\em With Definition \ref{def.MDS-NMDS} in mind, it is therefore of particular interest to 
construct families of non-GRS MDS-NMDS codes.}

\subsection{Our motivations for constructing non-GRS MDS-NMDS codes from deep holes}


Given an $[n,k]_q$ linear code $\C$ and any vector $\mathbf{u}\in \F_q^n$, 
the {\em error distance} of $\mathbf{u}$ to $\C$ is defined by  
$d(\mathbf{u},\C)=\min\{d(\mathbf{u},\mathbf{c}):~\mathbf{c}\in \C\}$. 
We call $$\rho(\C)=\max\{d(\mathbf{u}, \C):~\mathbf{u}\in \F_q^n\}$$ the {\em covering radius} of $\C$.  
Then the {\em deep holes} of $\C$ are just the vectors in $\F_q^n$ whose error distances to $\C$ achieve $\rho(\C)$. 
Define the set of deep holes of $\C$ by  
\begin{align}\label{eq.deep holes}
    D(\C)=\{\mathbf{u}\in \F_q^n:~ d(\mathbf{u},\C)=\rho(\C)\}.
\end{align}

It is well known that deep holes of linear codes are closely related to syndrome decoding, 
as they characterize extremal cosets attaining the covering radius and thus, correspond to 
the most difficult instances of the decoding problem. 
However, McLoughlin \cite{M1984} has proven that the computational difficulty of determining the covering radius 
of random linear codes strictly exceeds NP-completeness, 
which further implies that the problem of determining deep holes of general linear codes is also computationally intractable.
By analyzing the structure of deep hole cosets, parameters of extended codes, 
connections with problems and results in finite geometry, syndromes, subset sum problems, 
full automorphism groups, and character sums, a large body of work has focused on deep holes of  
first-order Reed-Muller codes, GRS codes, elliptic curve codes, dual codes of Roth-Lempel codes, 
(+)-twisted GRS codes and their extended codes in 
\cite{O2012,ZW2023,D1994,K2017,KW2016,LW2008,ZCL2016,FXZ2025,LZS2025,WDC2025} and the references therein. 
Given the inherent nature of these difficult problems, these results are also of independent interest in cryptanalysis.

Returning to coding theory itself, a recently emerging direction is to construct new linear codes with good parameters 
and desirable structural properties from linear codes possessing explicit deep holes. 
{A major advance was achieved by Wu \emph{et al.} \cite{WDC2025}, 
who established a clear structural connection between the MDS property of extended codes 
and the deep holes of the dual codes. 
Building on this result, Li \emph{et al.} \cite{LZS2025} further identified precise conditions 
under which the resulting codes are not monomially equivalent to GRS codes, 
while Yang \cite{Y2025} applied a similar approach to $3$-error-correcting binary primitive BCH codes. 
These ideas have been further developed and applied in a number of recent works 
\cite{ADV2025,LJL2025,LELL2025,ZYZ2026,WHLD2024}, 
showing that characterizations in terms of covering radii and deep holes are both conceptually clear 
and effective for systematic constructions of longer MDS codes.}

{In contrast, the study of extended codes in the NMDS setting has not yet reached 
a comparable level of structural understanding. 
A recent work by Ma \emph{et al.} \cite{MKZ2026} provided a necessary and sufficient condition 
(see Theorem~1 in \cite{MKZ2026} or Lemma~\ref{lem.MKZ} in this paper) for an extended code preserving the NMDS property. 
However, unlike the MDS case \cite{WDC2025}, this characterization does not use deep holes of NMDS codes. 
Instead, it requires verifying two conditions simultaneously, including checking the linear independence 
of certain column combinations and determining the intersection sizes between hyperplanes and specific point sets. 
This verification procedure involves an exhaustive search for the extending vector, 
which leads to exponential computational complexity 
(see Theorem~\ref{th.complexity_MKZ} and Remark~\ref{rem.complexity} for details). 
Moreover, it does not provide a theoretical guarantee for the existence of a suitable extending vector. 
This difference further motivates us to develop a characterization for the NMDS case in terms of deep holes.}

{Motivated by the success of previous results for MDS codes using deep holes \cite{WDC2025}, 
the high computational complexity of existing methods for NMDS codes that do not use deep holes \cite{MKZ2026}, 
and the important role of non-GRS MDS and NMDS codes, 
we are led to the following question.} 

\begin{question}\label{ques.deep hole}
{Can we develop a unified framework to explicitly construct more non-GRS MDS-NMDS codes from deep holes, 
thereby bypassing the complex exhaustive verification and providing a structural guarantee for the non-GRS MDS and NMDS properties?}
\end{question}

\subsection{Our contributions}

This paper gives an affirmative answer to Question~\ref{ques.deep hole} for codes 
with covering radius $n-k$ and explicit deep holes. 
Our main result establishes a unified framework for constructing new families of non-GRS MDS codes 
and NMDS codes with larger lengths and dimensions from known ones with explicit deep holes, 
thereby providing a solution to Question~\ref{ques.deep hole}.

\begin{theorem}\label{th.main1}
Let $\C$ be an $[n,k]_q$ linear code with generator matrix $G$ and covering radius $\rho(\C)=n-k$. 
Suppose that $\mathbf{u}\in \F_q^{n}$ is a deep hole of $\C$. 
Let $\C'$ be the linear code generated by
\[
\begin{pmatrix}
G & {\bf 0}\\ 
\mathbf{u} & 1
\end{pmatrix}.
\]
Then $\C'$ is an $[n+1,k+1,n-k+1]_q$ MDS (resp. an $[n+1,k+1,n-k]_q$ NMDS) code if and only if 
$\C$ is an $[n,k,n-k+1]_q$ MDS (resp. an $[n,k,n-k]_q$ NMDS) code. 
Moreover, if $\C$ is non-GRS, then so is $\C'$.
\end{theorem}

As stated in Corollary~\ref{coro.family}, by taking $\C$ to be a family of $[n,k]_q$ non-GRS MDS-NMDS codes with covering radius $n-k$, 
Theorem~\ref{th.main1} directly yields the existence of another family of $[n+1,k+1]_q$ non-GRS MDS-NMDS codes. 
Furthermore, our main contributions can be summarized as follows.

\begin{itemize}
\item[\rm 1)] We present an equivalent formulation of Theorem~\ref{th.main1} 
in terms of the second kind of extended codes \cite{SDC2024FFA,SDC2024DM} in Theorem~\ref{th.equivalent_framework}. 
As a special case, our result recovers a main construction of Wu \emph{et al.}~\cite{WDC2025} 
for generating additional MDS codes. 
In contrast to~\cite{WDC2025}, our framework further extends to the NMDS setting 
and accommodates the non-GRS property. 
Necessary conditions for the applicability of our main result are provided in 
Remark~\ref{rem.111}, Table~\ref{table.covering radius}, and Example~\ref{exam.counterexample}. 
Moreover, as discussed in Theorem \ref{th.complexity_MKZ} and Remark~\ref{rem.complexity}, 
Theorems~\ref{th.main1} and~\ref{th.equivalent_framework} 
lead to a provable reduction in computational complexity 
compared with the approach in~\cite{MKZ2026}. 

\item[\rm 2)] As an application of Theorem~\ref{th.main1}, we consider the specific family of $[n+1,k]_q$ 
non-GRS MDS-NMDS extended subcodes of GRS (ESGRS) codes $\C_k(\SSS,{\bf v},\infty)$ (see Definition~\ref{def.codes}) 
and derive new families of $[n+2,k+1]_q$ non-GRS MDS-NMDS codes. 
We show that the covering radius of $\C_k(\SSS,{\bf v},\infty)$ equals $n-k+1$ in Lemma~\ref{lem.covering radius},  
and characterize criteria for identifying deep holes of ESGRS codes in Theorems~\ref{th.deep hole} and~\ref{th.deep hole g(x)}.

\item[\rm 3)] We further characterize two classes of received vectors that form deep holes of $\C_k(\SSS,{\bf v},\infty)$ in 
Theorems~\ref{th.specific deep hole1} and~\ref{th.specific deep hole2}, respectively, 
by establishing connections with Roth-Lempel codes and elementary symmetric functions. 
Specifically, the first class corresponds to polynomials of the form
\[
g(x)=g_{k-1}x^{k-1}+f(x),\quad g_{k-1}\in \F_q^*,\quad f(x)\in \mathcal{V}_k,
\]
while the second class is associated with polynomials of the form
\[
g(x)=g_{k+1}x^{k+1}+g_{k-1}x^{k-1}+f(x),\quad 
g_{k+1}\in \F_q^*,\quad g_{k-1}\in \F_q,\quad f(x)\in \mathcal{V}_k, 
\]
where $\mathcal{V}_k=\{f(x)=\sum_{i=0}^{k-2}f_ix^i+f_kx^k:~f_i\in \F_q,~\forall~ i=0,1,\ldots,k-2,k\}.$

\item[\rm 4)] By applying the deep holes obtained in~3) together with the improved non-GRS MDS properties of ESGRS codes established in Theorem~\ref{th.non-GRS}, 
we construct three additional families of $[n+2,k+1]_q$ non-GRS MDS-NMDS codes in Theorems~\ref{th.more non-GRS MDS codes111} and~\ref{th.more non-GRS MDS codes222}. 
We further analyze their monomial equivalence with Roth-Lempel codes in Theorem~\ref{th.more non-GRS MDS codes111} and Remark~\ref{rem.answer_to_Problem1}.2). 
In particular, our results reveal the existence of a special subclass of Roth-Lempel codes that are non-GRS MDS-NMDS codes, as stated in Remark~\ref{rem.answer_to_Problem1}.1). 
For practical implementation, we provide an explicit algorithm for generating such non-GRS MDS-NMDS codes with larger lengths and dimensions in Algorithm~\ref{alg.1}. 
Several concrete examples are also included in Examples~\ref{exam.4} and~\ref{exam.5} to illustrate our constructions.
\end{itemize}

This paper is organized as follows. 
After the introduction, Section~\ref{sec2.preliminaries} reviews fundamental notation and preliminary results. 
Section~\ref{sec.3} addresses Question~\ref{ques.deep hole} by establishing a unified framework for constructing new families of non-GRS MDS-NMDS codes with larger lengths and dimensions from known ones with explicit deep holes. 
Section~\ref{sec.4} applies the framework developed in Section~\ref{sec.3} to derive further families of non-GRS MDS-NMDS codes with larger lengths and dimensions, and investigates the monomial equivalence between these newly constructed non-GRS MDS codes and Roth-Lempel codes. 
Finally, Section~\ref{sec.concluding remarks} concludes the paper.

\section{Preliminaries}\label{sec2.preliminaries}

Before proceeding, we first introduce some notation that will be used throughout this paper. 
For distinct elements $a_{1},a_2,\ldots,a_{n}\in \F_q$ and nonzero elements $v_1,v_2,\ldots,v_n\in \F_q$, 
we let $\SSS=\{a_{1},a_2,\ldots,a_{n}\}$ and ${\bf v}=(v_1,v_2,\ldots,v_n)$. 
For the subset $\SSS\subseteq \F_q$, we call $\SSS$ an {\em $(n,k,\delta)$-set} if 
$$|\SSS|=n~{\rm and}~a_{i_1}+a_{i_2}+ \cdots +a_{i_k}\neq \delta~
{\rm for~ any}~\{a_{i_1},a_{i_2},\ldots,a_{i_k} \} \subseteq \SSS.$$ 
Moreover, $\SSS$ is an $(n,k,0)$-set if and only if $\SSS$ is {\em $k$-zero-sum free}. 
Conversely, we say that {\em $\SSS$ contains a $k$-zero-sum subset} if $\SSS$ is not $k$-zero-sum free. 
For any two vectors ${\bf a}=(a_1,a_2,\ldots,a_n)$ and ${\bf b}=(b_1,b_2,\ldots,b_n)$ in $\F_q^n$, 
we define their {\em star product} by ${\bf a}\star {\bf b}=(a_1b_1,a_2b_2,\ldots,a_nb_n)$. 
Let ${\bf 0}$ (resp. ${\bf 1}$) be an appropriate {\em row} or {\em column vector} of all zeros (resp. ones). 
Let $\F_q[x]$ be the polynomial ring over $\F_q$. 
Denote by  
$$\F_{q}[x]_k=\left\{f(x)=\sum_{i=0}^{k-1}f_ix^i:~f_i\in \F_q,~\forall~ i=0,1,\ldots,k-1\right\}$$
and 
$$\mathcal{V}_k=\left\{f(x)=\sum_{i=0}^{k-2}f_ix^i+f_kx^k:~f_i\in \F_q,~\forall~ i=0,1,\ldots,k-2,k\right\}$$ 
two $k$-dimensional $\F_q$-linear spaces of polynomials. 

With the above notation, we can define Roth-Lempel codes \cite{RL1989} and extended subcodes of GRS codes \cite{ADV2025,LSZ2024}. 

\begin{definition}{\rm (\!\! \cite{RL1989})}\label{def.Roth-Lempel code}
    Let $k$ and $n$ be two positive integers satisfying $3\leq k\leq n\leq q$. 
    For any $\delta\in \F_q$, a {\em Roth-Lempel code}, 
    denoted by $\RL_{k,\delta}(\SSS)$, is an $[n+2,k]_q$ linear code defined by  
    \begin{align*}
        \begin{split}
            \RL_{k,\delta}(\SSS) =  \{(f(a_1),f(a_2),\ldots,f(a_n),f_{k-1}, 
            f_{k-2}+\delta f_{k-1}):~ f(x)\in \F_q[x]_k \},   
        \end{split}
    \end{align*}
    where $f_{k-2}$ and $f_{k-1}$ are the coefficients of $x^{k-2}$ and $x^{k-1}$ in $f(x)$, respectively. 
    Moreover, it has a generator matrix of the form  
    \begin{align}\label{eq.ESGRS_generator_matrix of Roth-Lempel code}
        G_{\RL_{k,\delta}(\SSS)}=\left(\begin{array}{cccccc}
            1 & 1 &  \ldots & 1 & 0 & 0 \\
            a_1 & a_2 & \ldots & a_n & 0 & 0 \\
            \vdots  & \vdots  & \ddots & \vdots & \vdots & \vdots \\
            a^{k-3}_1  & a^{k-3}_2  & \ldots & a^{k-3}_n & 0 & 0 \\
            a^{k-2}_1  & a^{k-2}_2 & \ldots & a^{k-2}_n & 0 & 1 \\
            a^{k-1}_1  & a^{k-1}_2 & \ldots & a^{k-1}_n & 1 & \delta \\
        \end{array}\right).
    \end{align}
\end{definition}

In \cite{RL1989}, Roth and Lempel derived a necessary and sufficient condition 
for a Roth-Lempel code to be non-GRS MDS.

\begin{lemma}{\rm (\!\! \cite[Page 655]{RL1989})}\label{lem.RL_code_MDS}
    For any $3\leq k\leq n\leq q$ and $\delta\in \F_q$,  
    $\RL_{k,\delta}(\SSS)$ is an $[n+2,k,n-k+3]_q$ non-GRS MDS code if and only if 
    $\SSS$ is an $(n,k-1,\delta)$-set.
\end{lemma}

\begin{definition}{\rm (\!\! \cite{ADV2025,LSZ2024})}\label{def.codes}
Let $n$ and $k$ be two positive integers satisfying $3\leq k\leq n-2\leq q-2$. 
An {\em extended subcode of GRS (ESGRS) code} associated with $\SSS$ and ${\bf v}$, denoted by $\C_k(\SSS,{\bf v},\infty)$
        is an $[n+1,k]_q$ linear code 
        \begin{align*}
            \C_k(\SSS,{\bf v},\infty)=
            \{(v_1f(a_1),v_2f(a_2),\ldots,v_nf(a_n),f_{k}): ~f(x)\in \mathcal{V}_k\},  
        \end{align*}
    where 
    $f_k$ is the coefficient of $x^k$ in $f(x)$. 
    It has a generator matrix of the form 
        \begin{align}\label{eq.ESGRS_generator_matrix}
            G_{\C_k(\SSS,{\bf v},\infty)}=\begin{pmatrix}    
        v_1 & v_2 &  \ldots & v_{n} & 0 \\ 
        v_1a_1 &  v_2a_2 &   \ldots & v_{n}a_{n} & 0 \\
        \vdots & \vdots &   \ddots & \vdots & \vdots \\
        v_1a_1^{k-2} &  v_2a_2^{k-2} &   \ldots & v_{n}a_{n}^{k-2} & 0 \\
        v_1a_1^k &  v_2a_2^k &  \ldots & v_{n}a_{n}^k & 1 \\
            \end{pmatrix} 
    \end{align}
\end{definition}

ESGRS codes were first introduced by Li $et~al.$ \cite{LSZ2024} and further developed by Abdukhalikov $et~al.$ \cite{ADV2025}. 
In \cite{LSZ2024}, it was shown that ESGRS codes form a family of MDS-NMDS codes, 
and necessary and sufficient conditions were established for an ESGRS code to be an MDS code or an NMDS code.
We rephrase this result in the following lemma. 

\begin{lemma}{\rm (\!\! \cite[Theorem 4]{LSZ2024})}\label{lem.must be MDS or NMDS} 
    ESGRS codes are a family of MDS-NMDS codes.
    Furthermore, the following statements hold. 
    \begin{enumerate}
        \item [\rm 1)] $\C_k(\SSS,{\bf v},\infty)$ is an $[n+1,k,n-k+2]_q$ MDS code if and only if $\SSS$ is $k$-zero-sum free. 
        \item [\rm 2)] $\C_k(\SSS,{\bf v},\infty)$ is an $[n+1,k,n-k+1]_q$ NMDS code if and only if $\SSS$ contains a $k$-zero-sum subset. 
    \end{enumerate}
\end{lemma}

For general linear codes, their MDS and NMDS properties can be characterized by their generator matrices.

\begin{lemma}{\rm (\!\!\cite[Lemma 7.3]{B2015})}\label{lem.MDS} 
    Let $\C$ be an $[n,k]_q$ linear code with generator matrix $G$. 
    Then $\C$ is MDS if and only if any $k$ columns of $G$ are linearly independent. 
\end{lemma}

\begin{lemma}{\rm (\!\!\cite[Lemma 3.1]{DL1995})}\label{lem.NMDS} 
    Let $\C$ be an $[n,k]_q$ linear code with generator matrix $G$. 
    Then $\C$ is NMDS if and only if $G$ satisfies the following three conditions:  
    \begin{enumerate}
        \item [\rm 1)] any $k-1$ columns of $G$ are linearly independent;
        \item [\rm 2)] there exist $k$ linearly dependent columns in $G$; 
        \item [\rm 3)] any $k+1$ columns of $G$ have full rank $k$.
    \end{enumerate} 
\end{lemma}

Finally, we recall the notions of punctured and shortened codes \cite{HP2003}. 
Let $\C$ be an $[n,k]_q$ linear code and let $1\leq i\leq n$. 
The \emph{punctured code} of $\C$ at the $i$th coordinate is defined as
\begin{align}\label{eq.puncture}
    \mathrm{Punc}_i(\C)=
\{(c_1,c_2,\ldots,c_{i-1},c_{i+1},\ldots,c_n):
(c_1,c_2,\ldots,c_n)\in \C\}.
\end{align}
The \emph{shortened code} of $\C$ at the $i$th coordinate is defined as
\begin{align}\label{eq.shorten}
    \mathrm{Short}_i(\C)=
\{(c_1,c_2,\ldots,c_{i-1},c_{i+1},\ldots,c_n):
(c_1,c_2,\ldots,c_n)\in \C,\ c_i=0\}.
\end{align}
With \eqref{eq.puncture} and \eqref{eq.shorten}, it is easy to check that $\mathrm{Punc}_i(\C^{\perp})=(\mathrm{Short}_i(\C))^{\perp}$. 
By \cite[Lemma 2]{LSEL2023}, punctured and shortened codes of an MDS code are still MDS. 
Moreover, by \cite[Theorem 10]{TCY2014}, for an NMDS code there exists a coordinate 
such that the corresponding punctured and shortened codes remain NMDS. 
However, puncturing and shortening operations may destroy the non-GRS property. 

\section{A framework for constructing non-GRS MDS-NMDS codes from deep holes}\label{sec.3} 

In this section, we develop a general framework to construct families of non-GRS MDS-NMDS codes using deep holes. 
We also provide an equivalent expression of this framework in terms of the so-called second kind of extended codes 
introduced by Sun $et~al.$ in \cite{SDC2024FFA,SDC2024DM} and recover a recent result presented by Wu $et~al.$ in \cite{WDC2025} 
as a special case.

\subsection{A unified framework}

Note that determining deep holes of a linear code requires prior knowledge of its covering radius. 
Throughout this section, we assume that the covering radius of the $[n,k]_q$ non-GRS MDS or NMDS code under consideration equals $n-k$. 
This assumption is reasonable since many known non-GRS MDS and NMDS codes indeed 
satisfy this condition \cite{FXZ2025,LJL2025,LZS2025,WDC2025}. 
We also provide a counterexample in Example \ref{exam.counterexample} to show that this condition 
is necessary for our framework.
Note also that Zhuang $et~al.$~\cite[Proposition~2]{ZCL2016} established an equivalent condition for a vector to be a deep hole of a Reed-Solomon code, 
based on the covering radius and a generator matrix of the code. 
More recently, Wu $et~al.$~\cite[Lemma~5]{WDC2025} removed the GRS restriction and extended this characterization to general MDS codes. 
However, since the codes studied in this paper may be NMDS, the results of~\cite[Lemma~5]{WDC2025} and~\cite[Proposition~2]{ZCL2016} are not directly applicable. 
This motivates the following lemma.

\begin{lemma}\label{lem.mds_nmds_deep_hole}
    Let $\C$ be an $[n,k]_q$ linear code with generator matrix $G$ and covering radius $\rho(\C)=n-k$. 
    Suppose that $\mathbf{u}\in \F_q^{n}$ and $\left(\begin{array}{c}
        G \\ 
        \mathbf{u}
    \end{array}\right)$ generates a linear code $\C_{\bf u}$. 
    If $d(\C)\geq \rho(\C)$, then $\mathbf{u}$ is a deep hole of $\C$ 
    if and only if $\C_{\bf u}$ is an $[n,k+1,n-k]_q$ MDS code. 
\end{lemma}
\begin{IEEEproof}
    Note that if $\mathbf{u}$ is a deep hole of $\C$, 
    then $\mathbf{u}\notin \C$, which implies that $\C_{\bf u}$ is an $[n,k+1]_q$ linear code. 
    By the assumption $d(\C)\ge \rho(\C)$ and the definition of the covering radius, 
we have
$
d(\C)\ge \rho(\C)=n-k\ge d({\bf u},\C).
$ 
Then 
\begin{align*}
        d(\C_{\bf u})=\min\{d(\C), d(\mathbf{u}, \C)\}=  d(\mathbf{u}, \C).
\end{align*}
    It further implies that $\mathbf{u}$ is a deep hole of $\C$ 
    if and only if $d(\C_{\bf u})=d(\mathbf{u}, \C)=\rho(\C)=n-k$, 
    if and only if $\C_{\bf u}$ is an $[n,k+1,n-k]_q$ MDS code. This completes the proof. 
\end{IEEEproof}

We are now ready to present the main result of this paper, 
which provides a general framework to get more non-GRS MDS (resp. NMDS) codes from 
known non-GRS MDS (resp. NMDS) codes with explicit deep holes.

\begin{theorem}\label{th.connection deep holes and extened codes}
    Let $\C$ be an $[n,k]_q$ linear code with generator matrix $G$ and covering radius $\rho(\C)=n-k$. 
    Suppose that $\mathbf{u}\in \F_q^{n}$ is a deep hole of $\C$. 
    Let $\C'$ be the linear code generated by the matrix $\left(\begin{array}{cc}
        G & {\bf 0}\\ 
        \mathbf{u} & 1
    \end{array}
     \right)$. 
    Then $\C'$ is an $[n+1,k+1,n-k+1]_q$ MDS (resp. an $[n+1,k+1,n-k]_q$ NMDS) code 
    if and only if $\C$ is an $[n,k,n-k+1]_q$ MDS (resp. an $[n,k,n-k]_q$ NMDS) code. 
     Moreover, if $\C$ is non-GRS, so is $\C'$. 
\end{theorem}
\begin{IEEEproof}
    Given the notation, we can express the code $\C'$ as 
    $$\C'=\{{\bf c}'=({\bf c}+a{\bf u}, a):~{\bf c}\in \C,~ a\in \F_q\}.$$
    If $a=0$, we have ${\bf c}'=({\bf c},0)$ for any ${\bf c}\in \C$ and $\wt({\bf c}')=\wt({\bf c})$.  
    Otherwise, we have ${\bf c}'=({\bf c}+a{\bf u}, a)$ for any ${\bf c}\in \C$ and $\wt({\bf c}')=\wt({\bf c}+a{\bf u})+1$. 
    From Lemma \ref{lem.mds_nmds_deep_hole}, we know that ${\bf c}+a{\bf u}$ is a codeword of 
    the $[n,k+1,n-k]_q$ MDS code $\C_{\bf u}$ 
    generated by $G_{\bf u}=\left(\begin{array}{c}
        G \\ 
        \mathbf{u}
    \end{array} \right)$. 
    Hence, we arrive at 
    \begin{align}\label{eq.mds_nmds}   
    d(\C')=\min\{d(\C), n-k+1\}=\left\{
        \begin{array}{ll}
            n-k+1, & \text{if } \C \text{ is MDS}, \\
            n-k, & \text{if } \C \text{ is NMDS}.
        \end{array}
    \right.
    \end{align}

    Note that $G'=\left(\begin{array}{cc}
        G & {\bf 0}\\ 
        \mathbf{u} & 1
    \end{array}
     \right)$ is a parity-check matrix of $\C'^{\perp}$. 
     Let $M$ be a $(k+1)\times k$ submatrix of $G'$. 
     In the following, we show that $M$ has full rank $k$ and hence, $d(\C'^{\perp})\geq k+1$.
     We have two cases to consider. 

    {\bf Case 1:} {\em All columns of $M$ are from the former $n$ columns of $G'$.} 
    Hence, $M$ is also a $(k+1)\times k$ submatrix of $G_{\bf u}$.  
    Since $\C_{\bf u}$ is an $[n,k+1,n-k]_q$ MDS code, it follows from Lemma \ref{lem.MDS} 
    that any $k+1$ columns of $G_{\bf u}$ are linearly independent. 
    It implies that $M$ has full rank $k$. 

    {\bf Case 2:} {\em $M$ contains the last column of $G'$.}  
    Suppose that all columns of $M$ are the $i_1$th, $i_2$th, $\ldots$, $i_{k-1}$th, 
    and $(n+1)$th columns of $G'$. 
    Denote by ${\bf g}_{i_1}, {\bf g}_{i_2}, \ldots, {\bf g}_{i_{k-1}}$, and $ {\bf g}_{n+1}$ 
    the corresponding column vectors of $G'$, where ${\bf g}_{i_j}=(g_{1,i_j},g_{2,i_j},\ldots,g_{k,i_j},u_{i_j})^T$ 
    for $1\leq j\leq k-1$ and ${\bf g}_{n+1}=(0,0,\ldots,0,1)^T$. 
    Assume that there exist $c_1,c_2,\ldots,c_{k}\in \F_q$ such that 
    \begin{align}\label{eq.rank}
       c_1{\bf g}_{i_1}+c_2{\bf g}_{i_2}+ \cdots +c_{k-1}{\bf g}_{i_{k-1}}+c_k{\bf g}_{n+1}={\bf 0} \Longleftrightarrow 
       \left\{
            \begin{array}{l}
                \sum_{j=1}^{k-1} c_j g_{m,i_j} = 0, \quad 1 \leq m \leq k, \\
                \sum_{j=1}^{k-1} c_j u_{i_j} +c_k = 0.
            \end{array}
        \right.
    \end{align}
    With Lemmas \ref{lem.MDS} and \ref{lem.NMDS}.1), we know that any $k-1$ columns of $G$ are linearly independent 
    whatever $\C$ is MDS or NMDS. 
    Therefore, the matrix 
    $$
    \left(
        \begin{array}{cccc}
            g_{1,i_1} & g_{1,i_2} & \ldots & g_{1,i_{k-1}} \\
            g_{2,i_1} & g_{2,i_2} & \ldots & g_{2,i_{k-1}} \\
            \vdots & \vdots & \ddots & \vdots \\
            g_{k,i_1} & g_{k,i_2} & \ldots & g_{k,i_{k-1}} \\
        \end{array}
    \right) 
    $$
    has full rank $k-1$, which implies that $c_1=c_2=\cdots=c_{k-1}=0$. 
    Then the last equation in the system of equations \eqref{eq.rank} gives that $c_k=0$. 
    As a result, $M$ also has full rank $k$. 

    Combining with the well-known Singleton bound, 
    we immediately conclude that $d(\C')\in \{n-k,~ n-k+1\}$ and 
    $d(\C'^{\perp})\in \{k+1,~ k+2\}$. 
    It then turns out from the fact that duals of MDS codes are also MDS and \eqref{eq.mds_nmds} 
    that $\C'$ is an $[n+1,k+1,n-k+1]_q$ MDS (resp. an $[n+1,k+1,n-k]_q$ NMDS) code 
    if and only if $\C$ is an $[n,k,n-k+1]_q$ MDS (resp. an $[n,k,n-k]_q$ NMDS) code. 

It then remains to prove the non-GRS property. 
Recall that
$
\C'=\{({\bf c}+a{\bf u},a):~{\bf c}\in\C,\ a\in\F_q\}.
$ 
By the definition of shortened codes in \eqref{eq.shorten}, 
shortening $\C'$ at the last coordinate gives
\[
\mathrm{Short}_{n+1}(\C')=
\{ {\bf c}+a{\bf u}:~({\bf c}+a{\bf u},a)\in \C',\ a=0\}=
\{ {\bf c}:~{\bf c}\in\C\}=\C.
\]
Suppose, to the contrary, that $\C'$ is monomially equivalent to a GRS or an EGRS code. 
Since the shortening operation is compatible with the monomial equivalence, 
$\mathrm{Short}_{n+1}(\C')$ is monomially equivalent to the shortened code of a GRS code or an EGRS code 
at the corresponding coordinate. 
Note also that shortened codes of GRS and EGRS codes are again GRS or EGRS codes up to the monomial equivalence. 
Therefore, $\C$ is monomially equivalent to a GRS or an EGRS code, contradicting the assumption that $\C$ is non-GRS. 
This completes the proof. 
\end{IEEEproof}

\begin{corollary}\label{coro.family}
    Given a family of $[n,k]_q$ non-GRS MDS-NMDS codes $\C$ with covering radius $\rho(\C)=n-k$,  
    there exists a family of $[n+1,k+1]_q$ non-GRS MDS-NMDS codes $\C'$. 
\end{corollary}
\begin{IEEEproof}
    Given the parameters $[n,k]_q$, both the vectors in $\F_q^{n}$ and their error distances to $\C$ are finite.
    With \eqref{eq.deep holes}, $D(\C)\neq \emptyset$ and thus, $\C$ must have deep holes. 
    For each non-GRS MDS (resp. NMDS) code in $\C$, we can arbitrarily select a deep hole and apply Theorem \ref{th.connection deep holes and extened codes} 
    to construct a new $[n+1,k+1]_q$ non-GRS MDS (resp. NMDS) code. 
    This completes the proof. 
\end{IEEEproof}

\begin{remark}\label{rem.111}
Although it is well known that shortening an $[n+1,k+1]_q$ MDS (resp.\ NMDS) code at a suitable coordinate 
can yield an $[n,k]_q$ MDS (resp. NMDS) code, the converse of Corollary~\ref{coro.family} does not necessarily hold, 
since the shortening operation may fail to preserve the non-GRS property.
\end{remark}

Note that throughout this paper we assume that the covering radius of the 
$[n,k]_q$ non-GRS MDS or NMDS code under consideration equals $n-k$. 
On the one hand, many families of codes have been shown to satisfy this condition 
\cite{FXZ2025,LJL2025,LZS2025,WDC2025}, which makes our framework broadly applicable. 
We list some known results in Table \ref{table.covering radius}. 
Note that deep holes of these codes have been partially determined in the corresponding references. 
Therefore, applying Theorem \ref{th.connection deep holes and extened codes} to these codes with explicit deep holes 
can directly yield more non-GRS MDS codes and NMDS codes with larger lengths and dimensions. 
We omit the details here for brevity.

\begin{table*}[ht!]
\centering
\caption{Some known families of $[N,K]_q$ non-GRS MDS and NMDS codes with covering radius $N-K$}\label{table.covering radius} 
\setlength{\tabcolsep}{8pt} 
\renewcommand{\arraystretch}{1.3} 
\newcolumntype{H}{>{\setbox0=\hbox\bgroup}c<{\egroup}}
\resizebox{\textwidth}{!}{
\begin{threeparttable}
\begin{tabular}{l|l|lH|l}
\toprule
Code family & Type & Condition & Deep hole  & Reference\\ 
\midrule
Twisted GRS $[n,k]_q$ \cite{BBPR2018} & MDS-NMDS\tnote{1}  & $2\leq k\leq n-1\leq q-1$ & xxx &  \cite[Theorem 4]{FXZ2025} \\

Dual of GRL $[n+2,k]_q$ \cite{LJL2025} & Non-GRS MDS & Two specific conditions\tnote{2} & \cite[Theorem 4.6]{LJL2025} & \cite[Theorem 4.6]{LJL2025} \\

Extended twisted GRS $[n+1,k]_q$ \cite{LZS2025} & Non-GRS MDS & $2\leq k\leq n-1\leq q-1$ & - & \cite[Theorem 5]{LZS2025} \\

Cyclic $[q+1,q-2u+2]_q$ \cite{SD2024} & Non-GRS MDS &  $u\in \{2, \frac{q+1}{2}\}$ & {\bf 1} &  \cite[Theorems 3 and 4]{WDC2025} \\
\bottomrule
\end{tabular}
\begin{tablenotes}
\footnotesize
\item[1] This family of twisted GRS codes could be monomially equivalent to GRS codes in some special cases (see \cite[Example 1]{LELL2025}).
\item[2] See Conditions (1) and (2) in \cite[Theorem 3.3]{LJL2025} for details.
\end{tablenotes}
\end{threeparttable}}
\end{table*}

On the other hand, the following example demonstrates that the assumption 
$\rho(\mathcal{C})=n-k$ is essential for our framework. 
Without this assumption, the resulting codes may fail to preserve the non-GRS MDS or NMDS properties.

\begin{example}\label{exam.counterexample}
    Let $\C$ be the binary extended Hamming code with generator matrix 
    $$G=\begin{pmatrix}
        1 & 0 & 0 & 0 & 0 & 1 & 1 & 1 \\
        0 & 1 & 0 & 0 & 1 & 0 & 1 & 1 \\
        0 & 0 & 1 & 0 & 1 & 1 & 0 & 1 \\
        0 & 0 & 0 & 1 & 1 & 1 & 1 & 0 \\
    \end{pmatrix}.$$
    It is well known that $\C$ is an $[8,4,4]_2$ NMDS code with covering radius $\rho(\C)=2$ (see also \cite[Example 1.12.4]{HP2003}), 
    without satisfying $\rho(\C)=n-k=4$. 
    Verified by Magma \cite{Magma}, the vector ${\bf u}=(0,1,1,1,0,1,0,0)$ is a deep hole of $\C$, 
    but the codes $\C_{\bf u}$ and $\C'$ with generator matrices  
    $$\left(\begin{array}{c}
        G \\ 
        \mathbf{u}
    \end{array}\right)=\begin{pmatrix}
        1 & 0 & 0 & 0 & 0 & 1 & 1 & 1 \\
        0 & 1 & 0 & 0 & 1 & 0 & 1 & 1 \\
        0 & 0 & 1 & 0 & 1 & 1 & 0 & 1 \\
        0 & 0 & 0 & 1 & 1 & 1 & 1 & 0 \\
        0 & 1 & 1 & 1 & 0 & 1 & 0 & 0
\end{pmatrix}~{\rm and}~\left(\begin{array}{cc}
        G & {\bf 0}\\ 
        \mathbf{u} & 1
    \end{array}
     \right)=\begin{pmatrix}
        1 & 0 & 0 & 0 & 0 & 1 & 1 & 1 & 0 \\
        0 & 1 & 0 & 0 & 1 & 0 & 1 & 1 & 0\\
        0 & 0 & 1 & 0 & 1 & 1 & 0 & 1 & 0 \\
        0 & 0 & 0 & 1 & 1 & 1 & 1 & 0 & 0 \\
        0 & 1 & 1 & 1 & 0 & 1 & 0 & 0 & 1
\end{pmatrix},$$
    as defined in Lemma \ref{lem.mds_nmds_deep_hole} and Theorem \ref{th.connection deep holes and extened codes},  
    have parameters $[8,5,2]_2$ and $[9,5,3]_2$, respectively. 
    Therefore, neither $\C_{\bf u}$ nor $\C'$ is NMDS.
\end{example}

\subsection{An equivalent framework in terms of the second kind of extended codes} 

In this subsection, we rephrase Theorem \ref{th.connection deep holes and extened codes} in terms of the so-called second kind of extended codes introduced 
by Sun $et~al.$ in \cite{SDC2024FFA,SDC2024DM}. 
Let ${\bf u}=(u_1,u_2,\ldots,u_n)\in \F_q^n$ be a vector. 
Sun $et~al.$ \cite{SDC2024FFA,SDC2024DM} introduced the so-called 
{\em second kind of extended code} $\overline{\C}({\bf u})$ of $\C$ defined by 
\begin{align}\label{eq.extended code}
    \overline{\C}({\bf u})=\left\{(c_1,c_2,\ldots,c_n, c_{n+1}):~ (c_1,c_2,\ldots,c_n)\in \C,~c_{n+1}=\sum_{i=1}^{n} c_i u_i\right\}. 
\end{align}

\begin{lemma}(\!\!\cite{SDC2024FFA})\label{lem.extended code}
Let $\C$ be an $[n,k,d]_q$ linear code and ${\bf u}\in \F_q^n$.
If $\C$ has generator matrix $G$ and parity-check matrix $H$, then the generator and parity-check
matrices for the extended code $\overline{\C}({\bf u})$ in~\eqref{eq.extended code} are 
\begin{align*}
    \overline{G}=\begin{pmatrix}
    G & G{\bf u}^{T}
    \end{pmatrix}~{\rm and}~
    \overline{H}=
    \begin{pmatrix}
    H & {\bf 0}\\
    {\bf u} & -1
    \end{pmatrix}. 
\end{align*}
\end{lemma}

Using the second kind of extended codes, we get the following equivalent expression of 
Theorem \ref{th.connection deep holes and extened codes}.

\begin{theorem}\label{th.equivalent_framework}
    Let $\C$ be an $[n,k]_q$ linear code with covering radius $\rho(\C)=n-k$. 
    Suppose that $\mathbf{u}\in \F_q^{n}$ is a deep hole of $\C$. 
    Then the extended code $\overline{\C^{\perp}}({\bf u})$ is an $[n+1,n-k,k+2]_q$ MDS (resp. an $[n+1,n-k,k+1]_q$ NMDS) code 
    if and only if $\C$ is an $[n,k,n-k+1]_q$ MDS (resp. an $[n,k,n-k]_q$ NMDS) code. 
     Moreover, if $\C$ is non-GRS, so is $\overline{\C^{\perp}}({\bf u})$. 
\end{theorem}
\begin{IEEEproof}
Suppose that $G$ is a generator matrix of $\C$. 
Since a generator matrix of a linear code serves as a parity-check matrix of its dual code, 
Lemma~\ref{lem.extended code} implies that $(\overline{\C^{\perp}}({\bf u}))^{\perp}$ admits the generator matrix
\[
\begin{pmatrix}
G & {\bf 0}\\
{\bf u} & -1
\end{pmatrix}.
\]
It follows that $(\overline{\C^{\perp}}({\bf u}))^{\perp}$ is monomially equivalent to the linear code $\C'$ 
defined in Theorem~\ref{th.connection deep holes and extened codes}. 
Therefore, the desired result follows directly from Theorem~\ref{th.connection deep holes and extened codes}, 
together with the facts that monomial equivalence preserves duality and that dual codes of GRS and EGRS codes 
are again monomially equivalent to GRS and EGRS codes, respectively.
\end{IEEEproof}

\begin{remark}
Let $\C$ be an $[n,k,n-k+1]_q$ MDS code. 
In \cite[Theorem~1]{WDC2025}, Wu $et~al.$ proved that 
for any ${\bf u}\in (\F_q)^n$, the extended code $\overline{\C}({\bf u})$ defined in \eqref{eq.extended code} 
is MDS if and only if $\rho(\C^{\perp})=k$ and ${\bf u}$ is a deep hole of the dual code $\C^{\perp}$. 
By taking duals in Theorem~\ref{th.connection deep holes and extened codes} and observing that the dual 
of an MDS code is again MDS, one can likewise obtain $[n+1,n-k,k+2]_q$ MDS codes from their result. 
Therefore, Theorem~\ref{th.equivalent_framework} recovers \cite[Theorem~1]{WDC2025} as a special case of our equivalent framework 
in the sense that it also yields additional MDS codes. 
Moreover, our approach also extends the discussion to the NMDS case 
and incorporates the non-GRS property.
\end{remark}

At the end of this section, we compare Theorems~\ref{th.connection deep holes and extened codes} and~\ref{th.equivalent_framework} 
with \cite[Theorem 1]{MKZ2026} in order to construct more NMDS codes. 
For completeness, we recall \cite[Theorem 1]{MKZ2026} below.

\begin{lemma}{\rm (\!\! \cite[Theorem 1]{MKZ2026})}\label{lem.MKZ}
    Let $\C$ be an $[n,k]_q$ NMDS code with generator matrix $G$ and $2\leq k\leq n-2$. 
    Let $S$ be the set of column vectors of $G$ and 
    let $\mathcal{H}_{\bf v}=\{{\bf x}\in \F_q^{k}:~\langle {\bf x}, {\bf v} \rangle_{\rm E}=0\}$ be 
    the hyperplane over $\F_{q}$ associated with the vector ${\bf v}\in \F_{q}^{k}$.  
    Then for any nonzero vector ${\bf u}\in \F_q^n$, the extended code 
    $\overline{\C}({\bf u})$ remains NMDS if and only if the following conditions are satisfied:
    \begin{enumerate}
        \item Any $k-2$ columns of $G$ together with $G{\bf u}^T$ are linearly independent; 
        \item For any nonzero vector ${\bf v}\in \F_q^k$, if ${\bf v}G{\bf u}^T=0$, then $|\mathcal{H}_{\bf v}\cap S|\leq k-1$.  
    \end{enumerate}
\end{lemma}

\begin{theorem}\label{th.complexity_MKZ}
Let $\C$ be an $[n,k]_q$ NMDS code with generator matrix $G$ and $2\leq k\leq n-2$. 
For a fixed nonzero vector ${\bf u}\in \F_q^n$, verifying whether the extended code 
$\overline{\C}({\bf u})$ is NMDS by using the criterion in Lemma~\ref{lem.MKZ} 
requires at most
\[
O\left(\binom{n}{k-2}k(k-1)^2+q^{k-1}nk\right)
\]
field operations. Moreover, a direct exhaustive search over all nonzero 
${\bf u}\in\F_q^n$ has complexity at most
\begin{align}\label{eq.upper-bound}
O\left(q^n\left(\binom{n}{k-2}k(k-1)^2+q^{k-1}nk\right)\right).
\end{align}
\end{theorem}
\begin{IEEEproof}
Fix a nonzero vector ${\bf u}\in \F_q^n$ and put 
${\bf w}=G{\bf u}^T\in \F_q^k.$
We verify the two conditions in Lemma~\ref{lem.MKZ}.

For Condition~1), one needs to check whether any $k-2$ columns of $G$ together with 
${\bf w}$ are linearly independent. There are $\binom{n}{k-2}$ choices of such 
$k-2$ columns. For each choice, this amounts to checking the rank of a 
$k\times (k-1)$ matrix. By Gaussian elimination, this requires at most 
$O(k(k-1)^2)$ field operations. Hence, the cost of verifying Condition~1) is at most
\[
O\left(\binom{n}{k-2}k(k-1)^2\right).
\]

For Condition~2), we distinguish two cases.
\begin{itemize}
    \item If ${\bf w}={\bf 0}$, then Condition~1) already fails, since the set consisting of 
${\bf w}$ together with any $k-2$ columns of $G$ is linearly dependent. Hence, the 
verification terminates after Condition~1).

    \item If ${\bf w}\neq {\bf 0}$, then the vectors ${\bf v}\in \F_q^k$ satisfying 
$\langle {\bf v}, {\bf w} \rangle_{\rm E}=0$ form a hyperplane of dimension $k-1$, 
which contains $q^{k-1}$ vectors, among which $q^{k-1}-1$ are nonzero. 
By the definitions of $\mathcal H_{\bf v}$ and $S$, 
for each such nonzero ${\bf v}$, computing $|\mathcal H_{\bf v}\cap S|$ requires testing all $n$ 
column vectors of $G$. Each test is an inner product of two vectors in $\F_q^k$, 
thereby costing $O(k)$ field operations. It follows that the cost of verifying Condition~2) 
is at most
\[
O((q^{k-1}-1)nk)=O(q^{k-1}nk).
\]
\end{itemize}

Combining the two estimates, the verification complexity for one fixed nonzero 
${\bf u}$ is at most
\[
O\left(\binom{n}{k-2}k(k-1)^2+q^{k-1}nk\right).
\]
Finally, a direct exhaustive search over all nonzero vectors in $\F_q^n$ contains 
$q^n-1$ candidates. Multiplying the above verification cost by $q^n-1$ and absorbing 
lower-order terms yields the asserted upper bound in \eqref{eq.upper-bound}. 
This completes the proof. 
\end{IEEEproof}

\begin{remark}\label{rem.complexity}
We compare the direct exhaustive implementation of the criterion in \cite[Theorem 1]{MKZ2026} with our framework in 
Theorems~\ref{th.connection deep holes and extened codes} and~\ref{th.equivalent_framework} for constructing longer NMDS codes.

\begin{enumerate}
\item Theorem~\ref{th.complexity_MKZ} shows that the criterion in \cite[Theorem 1]{MKZ2026} 
leads to an exhaustive verification procedure whose complexity is exponential in the code length $n$, 
since it involves a search over $q^n-1$ candidate vectors in $\F_q^n$. 
Passing to the dual code replaces the inner verification cost by 
\begin{align}\label{eq.upper-bound222}
O\left(q^n\left(\binom{n}{k+2}(n-k)(n-k-1)^2 + q^{n-k-1}n(n-k)\right)\right),
\end{align}
which may reduce the dependence on $k$ when $k$ is close to $n$, 
but does not eliminate the outer exhaustive search over $\F_q^n$. 
Thus, the overall procedure remains exponential in $n$.

\item For many known classes of NMDS codes, the admissible length $n$ is bounded in terms 
of the field size $q$. It has been conjectured that $n \le q + 2\sqrt{q}$ for $q$-ary NMDS codes in \cite{DL1995,LR2015}. 
Hence, constructing longer codes typically requires working over larger finite fields. 
In such regimes, the factor $q^n$ in \eqref{eq.upper-bound} and \eqref{eq.upper-bound222} 
becomes increasingly prohibitive. 
When $q$ grows linearly with $n$, we have $q^n = n^{\Theta(n)}$, which grows faster than any exponential function of the form $c^n$. 
This further highlights the computational infeasibility of \cite[Theorem 1]{MKZ2026}.

\item Since the criterion in \cite[Theorem 1]{MKZ2026} depends on ${\bf u}$ only through 
${\bf w}=G{\bf u}^T$, one may also reduce the outer enumeration from $\F_q^n$ to 
$\F_q^k$ in \eqref{eq.upper-bound}, or to $\F_q^{n-k}$ in \eqref{eq.upper-bound222} 
after passing to the dual code. Nevertheless, the resulting method remains 
exponential and verification-based. 
In contrast, our framework in Theorems~\ref{th.connection deep holes and extened codes} 
and~\ref{th.equivalent_framework} uses a deep hole of the base code as an explicit input. 
Once such a deep hole is available, the longer code is obtained simply by adjoining the 
corresponding row to a generator matrix, and no outer exhaustive search over arbitrary 
extending vectors is required. This is particularly useful for structured families of 
non-GRS MDS and NMDS codes for which explicit classes of deep holes are already known or 
can be described by polynomial evaluations, as illustrated by the families listed in 
Table~\ref{table.covering radius} (see the corresponding references for their deep holes) 
and by the family of ESGRS codes studied in Section~\ref{sec.4}. 
Therefore, our approach does not attempt to solve the general hard problem of determining 
deep holes of arbitrary linear codes \cite{M1984}. Instead, it separates this problem from 
the extension step and turns concrete deep holes into explicit constructions of longer 
non-GRS MDS and NMDS codes. 
\end{enumerate}
\end{remark}

\section{Application to ESGRS codes}\label{sec.4}

We have listed several known families of $[n,k]_q$ non-GRS MDS or NMDS codes with covering radius $n-k$ in Table \ref{table.covering radius}. 
Although explicit deep holes are known for these families listed in Table~\ref{table.covering radius}, 
they do not always provide the specific non-GRS MDS-NMDS family 
needed for the applications pursued in Question~\ref{ques.deep hole}. 
In this section, we apply the framework established in Theorem \ref{th.connection deep holes and extened codes} 
to the specific family of non-GRS MDS-NMDS ESGRS codes (see Definition \ref{def.codes}) 
introduced by Li $et~al.$ in \cite{LSZ2024} and further developed by Abdukhalikov $et~al.$ in \cite{ADV2025}. 
We determine the covering radius and deep holes of ESGRS codes, 
and further derive more non-GRS MDS-NMDS codes with larger lengths and dimensions from ESGRS codes with explicit deep holes.

\subsection{The covering radius and some criteria on deep holes of ESGRS codes}

We first need to determine the covering radius of ESGRS codes.

\begin{lemma}\label{lem.covering radius}
    For any $[n+1,k]_q$ ESGRS code $\C_k(\SSS,{\bf v},\infty)$, we have   
    $$\rho(\C_k(\SSS,{\bf v},\infty))=n-k+1.$$
\end{lemma}
\begin{IEEEproof}
    With Definition \ref{def.codes}, it follows from the redundancy bound (\!\! \cite[Corollary 11.1.3]{HP2003}) that 
    \begin{align}\label{eq.coveringradius11}
        \rho(\C_k(\SSS,{\bf v},\infty))\leq n-k+1.
    \end{align}
    On the other hand, by \eqref{eq.GRS.generator matrix} and \eqref{eq.ESGRS_generator_matrix}, 
    we have that $\C_k(\SSS,{\bf v},\infty)\subseteq \GRS_{k+1}(\SSS,{\bf v},\infty)$, 
    where $\GRS_{k+1}(\SSS,{\bf v},\infty)$ is an $[n+1,k+1,n-k+1]_q$ MDS code.  
    Combining with the supercode lemma (\!\! \cite[Corollary 11.1.5]{HP2003}), 
    it then implies that 
    \begin{align}\label{eq.coveringradius22}
        \begin{split}
         \rho(\C_k(\SSS,{\bf v},\infty))  
        \geq \min \{\wt(\mathbf{c}):\mathbf{c}\in \GRS_{k+1}(\SSS,{\bf v},\infty)\setminus \C_k(\SSS,{\bf v},\infty)\} 
        \geq n-k+1.
    \end{split} 
    \end{align}
    We complete the proof by considering \eqref{eq.coveringradius11} and \eqref{eq.coveringradius22} together. 
\end{IEEEproof}

Recall that the star product of two vectors $\mathbf{a}=(a_1,a_2,\ldots,a_n)$ and $\mathbf{b}=(b_1,b_2,\ldots,b_n)$ 
in $\F_q^n$ yields a new vector $\mathbf{a}\star \mathbf{b}=(a_1b_1,a_2b_2,\ldots,a_nb_n)$. 
With Lemmas \ref{lem.must be MDS or NMDS} and \ref{lem.covering radius}, 
we obtain a criterion for determining deep holes of ESGRS codes.

\begin{theorem}\label{th.deep hole}
    Suppose that $\mathbf{u}\in \F_q^{n+1}$ and $\left(\begin{array}{c}
        G_{\C_k(\SSS,{\bf v},\infty)} \\
        \mathbf{u}
    \end{array}\right)$ generates a linear code $\C_{\bf u}$. 
    Then $\mathbf{u}$ is a deep hole of $\C_k(\SSS,{\bf v},\infty)$ if and only if $\C_{\bf u}$ is an $[n+1,k+1,n-k+1]_q$ MDS code. 
\end{theorem}
\begin{IEEEproof}
    With Lemmas \ref{lem.must be MDS or NMDS} and \ref{lem.covering radius}, we have 
    $$d(\C_k(\SSS,{\bf v},\infty))\geq n-k+1=\rho(\C_k(\SSS,{\bf v},\infty))$$ for any ESGRS code $\C_k(\SSS,{\bf v},\infty)$.
    Then the desired result follows directly from Lemma \ref{lem.mds_nmds_deep_hole}. 
\end{IEEEproof}


{
\begin{theorem}\label{th.deep hole g(x)}
    Suppose that $\mathbf{g}=(g(a_1),g(a_2),\ldots,g(a_{n}),u_{n+1})\in \F_q^{n+1}$, 
    where $g(x)=g'(x)+f(x)\in \F_q[x]$ with $f(x)\in \mathcal{V}_k$. 
    If $(\mathbf{v},v_{n+1})\star \mathbf{g}$ is a deep hole of 
    the ESGRS code $\C_k(\SSS,{\bf v},\infty)$ for $v_{n+1}\in \F_q$,  
    then $g'(x)\in \F_q[x]\setminus \mathcal{V}_k$. 
\end{theorem}
}
\begin{IEEEproof}
    Assume that $g'(x)\in \mathcal{V}_k$. 
    Then $g(x)=g'(x)+f(x)\in \mathcal{V}_k$ and  
    $$(\mathbf{v}, v_{n+1})\star \mathbf{g}=(v_1(g'(a_1)+f(a_1)),v_2(g'(a_2)+f(a_2)),\ldots,v_n(g'(a_n)+f(a_n)),u_{n+1}v_{n+1}).$$
   It can be checked that 
        \begin{align*}
            \left(\begin{array}{c}
                G_{\C_k(\SSS,{\bf v},\infty)} \\ 
                (\mathbf{v},v_{n+1})\star \mathbf{g}
            \end{array}\right) = \left(\begin{array}{ccccc}
                v_1 & v_2 &  \ldots & v_{n} & 0 \\ 
                v_1a_1  & v_2a_2 &  \ldots & v_{n}a_{n} & 0 \\
                \vdots &  \vdots &   \ddots & \vdots & \vdots \\ 
                v_1a_1^{k-2}  & v_2a_2^{k-2} &   \ldots & v_{n}a_{n}^{k-2} & 0 \\ 
                v_1a_1^k  & v_2a_2^k &  \ldots & v_{n}a_{n}^k & 1 \\ 
                v_1(g'(a_1)+f(a_1)) & v_2(g'(a_2)+f(a_2))&  \ldots & v_{n}(g'(a_n)+f(a_n)) & u_{n+1}v_{n+1}
            \end{array}\right)  
        \end{align*}
        and it can be transformed to the matrix 
        \begin{align}\label{eq.matrix1_deephole}
            \left(\begin{array}{ccccc}
                v_1 & v_2 & \ldots & v_{n} & 0 \\ 
                v_1a_1  & v_2a_2 &  \ldots & v_{n}a_{n} & 0 \\
                \vdots &  \vdots &   \ddots & \vdots & \vdots \\ 
                v_1a_1^{k-2}  & v_2a_2^{k-2} &   \ldots & v_{n}a_{n}^{k-2} & 0 \\ 
                v_1a_1^k &  v_2a_2^k &  \ldots & v_{n}a_{n}^k & 1 \\ 
                0 & 0 & \ldots & 0 & u_{n+1}v_{n+1}-f_k-g_k 
            \end{array}\right) 
        \end{align}
        by a series of elementary row operations, where $f_k$ and $g_k$ are respective coefficients 
        of $x^k$ in $f(x)$ and $g'(x)$. 
        From Definition \ref{def.codes}, we have $k+1\leq n$. 
        It then follows that the determinants of all the $(k+1)\times (k+1)$ matrices consisting of any $k+1$ columns 
        of the first $n$ columns of the matrix given in \eqref{eq.matrix1_deephole} equal zero. 
        With Lemma \ref{lem.MDS}, we immediately conclude that 
        the linear code generated by $\left(\begin{array}{c}
            G_{\C_k(\SSS,{\bf v},\infty)} \\ 
            (\mathbf{v},v_{n+1})\star \mathbf{g}
        \end{array}\right)$ is not MDS. 
        According to Theorem \ref{th.deep hole},  
        we get that $g'(x)\in \F_q[x]\setminus \mathcal{V}_k$ if 
        $(\mathbf{v},v_{n+1})\star \mathbf{g}$ is a deep hole of 
        the ESGRS code $\C_k(\SSS,{\bf v},\infty)$. 
        This completes the proof. 
\end{IEEEproof}

\subsection{Two explicit classes of deep holes of ESGRS codes}\label{subsec.deep-holes}

With Theorem \ref{th.deep hole g(x)}, we can always assume that no term in $g'(x)$ belongs to $\mathcal{V}_k$. 
In the following, we consider two cases of $g'(x)$ in terms of their lowest degrees. 
In other words, we consider $g'(x)=g_{k-1}x^{k-1}$ with $g_{k-1}\in \F_q^*$ and 
$g'(x)=g_{k+1}x^{k+1}+g_{k-1}x^{k-1}$ with $g_{k+1}\in \F_q^*$ and $g_{k-1}\in \F_q$ in the following. 
For the first case, we establish a connection with the well-known Roth-Lempel codes.

{
\begin{theorem}\label{th.specific deep hole1}
    Suppose that $\mathbf{g}=(g(a_1),g(a_2),\ldots,g(a_{n}),u_{n+1})\in \F_q^{n+1}$, where  
    $g(x)=g_{k-1}x^{k-1}+f(x)\in \F_q[x]$ with $g_{k-1}\in \F_q^*$ and $f(x)\in \mathcal{V}_k$. 
    Then for any $v_{n+1}\in \F_q$, $(\mathbf{v},v_{n+1})\star \mathbf{g}$ is a deep hole of 
    the ESGRS code $\C_k(\SSS,{\bf v},\infty)$ 
    if and only if one of the following conditions holds: 
    \begin{enumerate}
        \item [\rm 1)]  $u_{n+1}v_{n+1}=f_k$; 
        \item [\rm 2)]  $u_{n+1}v_{n+1}\neq f_k$ and $\SSS$ is an $(n,k,g_{k-1}(u_{n+1}v_{n+1}-f_k)^{-1})$-set. 
    \end{enumerate}
\end{theorem}
}
\begin{IEEEproof}
Since $(\mathbf{v}, v_{n+1})\star \mathbf{g}=(v_1g(a_1),v_2g(a_2),\ldots,v_ng(a_n),u_{n+1}v_{n+1})$, 
we have that 
    \begin{align}\label{eq.matrix2_deephole}
        \left(\begin{array}{c}
            G_{\C_k(\SSS,{\bf v},\infty)} \\
            (\mathbf{v},v_{n+1})\star \mathbf{g}
        \end{array}\right) \overset{r}{\backsimeq}  \left(\begin{array}{ccccc}
        v_1 & v_2 &  \ldots & v_{n} & 0 \\ 
        v_1a_1  & v_2a_2 &  \ldots & v_{n}a_{n} & 0 \\
        \vdots &  \vdots &   \ddots & \vdots & \vdots \\
        v_1a_1^{k-2}  & v_2a_2^{k-2} &   \ldots & v_{n}a_{n}^{k-2} & 0 \\
        v_1a_1^{k-1}  & v_2a_2^{k-1} &   \ldots & v_{n}a_{n}^{k-1} & g_{k-1}^{-1}\left(u_{n+1}v_{n+1}-f_{k}\right) \\
        v_1a_1^k & v_2a_2^k&  \ldots & v_{n}a_{n}^k & 1 \\
    \end{array}\right),
\end{align}  
where ``$\overset{r}{\backsimeq}$" denotes that these two matrices are row equivalent. 
We consider the following two cases. 
    
{\textbf{Case 1:}} $u_{n+1}v_{n+1}=f_k$. Then the linear code generated by the matrix given in \eqref{eq.matrix2_deephole} 
        is an $[n+1,k+1,n-k+1]_q$ MDS code, actually the EGRS code $\GRS_{k+1}(\SSS,\mathbf{v},\infty)$. 
        According to Theorem \ref{th.deep hole}, $(\mathbf{v},v_{n+1})\star \mathbf{g}$ is a deep hole of $\C_k(\SSS,{\bf v},\infty)$. 

{\textbf{Case 2:}} $u_{n+1}v_{n+1}\neq f_k$. We note that the linear code generated by the matrix given in \eqref{eq.matrix2_deephole} is  
        equivalent to the linear code $\C'$ generated by the matrix 
        \begin{align*}
            \left(\begin{array}{ccccc}
                1 & 1 &  \ldots & 1 & 0 \\ 
                a_1 &  a_2 &  \ldots & a_{n} & 0 \\
                \vdots &  \vdots &   \ddots & \vdots & \vdots \\
                a_1^{k-2} &  a_2^{k-2} &   \ldots & a_{n}^{k-2} & 0 \\
                a_1^{k-1} &  a_2^{k-1} &   \ldots & a_{n}^{k-1} & 1 \\
                a_1^k &  a_2^k&  \ldots & a_{n}^k & g_{k-1}\left(u_{n+1}v_{n+1}-f_{k}\right)^{-1} \\
            \end{array}\right).
        \end{align*}
        Note also that $\C'$ coincides with a punctured code of the $[n+2,k+1]_q$ Roth-Lempel code  
        $\mathrm{RL}_{k+1,\delta}(\SSS)$ with $\delta=g_{k-1}\left(u_{n+1}v_{n+1}-f_{k}\right)^{-1}$, 
        defined as in Definition \ref{def.Roth-Lempel code}, by deleting its penultimate coordinate.  
        It then follows from Lemma \ref{lem.RL_code_MDS} and the fact that a punctured code of an MDS code is still MDS that 
        $\C'$ is an $[n+1,k+1,n-k+1]_q$ MDS code if $\SSS$ is an $\left(n,k,g_{k-1}(u_{n+1}v_{n+1}-f_k)^{-1}\right)$-set. 
        
        Conversely, if $\C'$ is an $[n+1,k+1,n-k+1]_q$ MDS code with generator matrix $G'$, 
        then the linear code generated by the matrix 
        $$\overline{G'}=\left(G',~~(0,0,\ldots,0,1)^T\right),$$ 
        that is, the matrix obtained by appending a column vector $(0,0,\ldots,0,1)^T$ of length $k+1$ to $G'$,  
        is necessarily MDS and this linear code is just the Roth-Lempel code $\RL_{k+1,\delta}(\SSS)$, 
        where $\delta=g_{k-1}\left(u_{n+1}v_{n+1}-f_{k}\right)^{-1}$. 
        Again, from Lemma \ref{lem.RL_code_MDS}, we conclude that $\SSS$ is 
        an $\left(n,k,g_{k-1}(u_{n+1}v_{n+1}-f_k)^{-1}\right)$-set. 
        Hence, $(\mathbf{v},v_{n+1})\star \mathbf{g}$ is a deep hole of $\C_k(\SSS,{\bf v},\infty)$ 
        if and only if 
        $\SSS$ is an $(n,k,g_{k-1}(u_{n+1}v_{n+1}-f_k)^{-1})$-set.
    
    Combining {\bf Cases 1} and {\bf 2}, we complete the proof. 
\end{IEEEproof}

For the second case, we use elementary symmetric functions to characterize the deep holes of ESGRS codes. 
Let $\sigma_{i}(x_1,x_2,\ldots,x_n)$ be the $i$th {\em elementary symmetric function} 
    in $n$ variables $x_1,x_2,\ldots,x_n$ such that 
    \begin{align*}
        \sigma_{i}(x_1,x_2,\ldots,x_n)
         =\left\{\begin{array}{ll}
            1, & {\rm if}~i=0, \\ 
            \sum_{I \subseteq \{1,2,\ldots,n\},~ |I|=i}\prod_{j\in I}x_j, & {\rm if}~1\leq i \leq n,\\
            0, & {\rm if}~i<0~{\rm or}~i>n.    
        \end{array} \right.
    \end{align*}
Furthermore, for any set $\mathcal{S}_k=\{t_1,t_2,\ldots,t_k\} \subseteq \mathcal{S}\subseteq \F_q$, we let  
$$\sigma_{i}(\mathcal{S}_k)=\sigma_{i}(t_1,t_2,\ldots,t_k).$$ 
Note also that $|\mathcal{S}_k|=k$ in this case. 

With this notation, we have the following lemma. 
\begin{lemma}{\rm (\!\! \cite[Lemma 2.3]{YZ2024})}\label{lem.det of generalized Vandermone matrix}
    Let $m$ be a fixed positive integer and let 
    $I_m=\left\{0,1,\ldots,m-1 \right\}=\left\{t_1,t_2,\ldots,t_s\right\}\cup \left\{r_1,r_2,\ldots,r_{s'}\right\}$ 
    {\color{black}be} any partition of $I_m$ with $m=s+s'$, $0=t_1<t_2<\cdots<t_s=m-1$ and $r_1<r_2<\cdots<r_{s'}$.  
    Then for any $\SSS=\left\{a_1,a_2,\ldots,a_s \right\}\subseteq \F_q$, we have  
    \begin{align*}
        \det \left(\begin{array}{cccc}
            a_1^{t_1} & a_2^{t_1} &  \ldots & a_s^{t_1} \\
            a_1^{t_2} & a_2^{t_2} &  \ldots & a_s^{t_2} \\
            \vdots & \vdots & \ddots & \vdots \\
            a^{t_s}_1 & a^{t_s}_2 &  \ldots & a^{t_s}_s \\
        \end{array}\right)=\prod_{1\leq i<j\leq s}(a_j-a_{i}) \cdot \det\left(\begin{array}{cccc}
        \sigma_{s-r_1}(\SSS)  & \sigma_{s-r_2}(\SSS)  &  \ldots & \sigma_{s-r_{s'}}(\SSS) \\
        \sigma_{s-r_1+1}(\SSS)  & \sigma_{s-r_2+1}(\SSS)  & \ldots & \sigma_{s-r_{s'}+1}(\SSS) \\
        \vdots & \vdots &  \ddots & \vdots \\
        \sigma_{s-r_1+s'-1}(\SSS)  & \sigma_{s-r_2+s'-1}(\SSS)  & \ldots & \sigma_{s-r_{s'}+s'-1}(\SSS) \\
    \end{array}\right).   
    \end{align*}
\end{lemma}

\begin{theorem}\label{th.specific deep hole2}
    Suppose that $\mathbf{g}=(g(a_1),g(a_2),\ldots,g(a_{n}),u_{n+1})\in \F_q^{n+1}$, where  
    $g(x)=g_{k+1}x^{k+1}+g_{k-1}x^{k-1}+f(x)\in \F_q[x]$ with $g_{k+1}\in \F_q^*$, $g_{k-1}\in \F_q$,  
    and $f(x)\in \mathcal{V}_k$. 
    Then for any $v_{n+1}\in \F_q$, $(\mathbf{v},v_{n+1})\star \mathbf{g}$ is a deep hole of 
    the ESGRS code $\C_k(\SSS,{\bf v},\infty)$ 
    if and only if 
    \begin{align*}
        g_{k-1}\notin 
        \left\{g_{k+1}\left(\sigma_2(\SSS_{k})-\sigma_1(\SSS_{k})^2\right)+(u_{n+1}v_{n+1}-f_k)\sigma_1(\SSS_{k}):~\SSS_{k} \subseteq \SSS   \right\} 
        \cup \left\{g_{k+1}\sigma_2(\SSS_{k+1}):~ \SSS_{k+1} \subseteq \SSS   \right\}.
    \end{align*}
\end{theorem}
\begin{IEEEproof}
Under the given condition, we know that 
        $\left(\begin{array}{c}
            G_{\C_k(\SSS,{\bf v},\infty)} \\
            (\mathbf{v},v_{n+1})\star \mathbf{g}
        \end{array}\right)$ 
        is row equivalent to the following matrix 
        \begin{align}\label{eq.matrix3_deephole}
            \left(\begin{array}{ccccc}
                v_1 & v_2 &  \ldots & v_{n} & 0 \\ 
                v_1a_1  & v_2a_2 &  \ldots & v_{n}a_{n} & 0 \\
                \vdots &  \vdots & \ddots & \vdots & \vdots \\
                v_1a_1^{k-2}  & v_2a_2^{k-2} &  \ldots & v_{n}a_{n}^{k-2} & 0 \\
                v_1a_1^k & v_2a_2^k &  \ldots & v_{n}a_{n}^k & 1 \\
                v_1(g_{k+1}a_1^{k+1}+g_{k-1}a_1^{k-1}) & v_2(g_{k+1}a_2^{k+1}+g_{k-1}a_2^{k-1}) &  \ldots & v_{n}(g_{k+1}a_n^{k+1}+g_{k-1}a_n^{k-1}) & u_{n+1}v_{n+1}-f_{k} \\
        \end{array}\right)
        \end{align}
    Let $M_{k+1}$ be any $(k+1)\times (k+1)$ submatrix of the matrix given in \eqref{eq.matrix3_deephole}. 
    Since $k+1\leq n$, to determine whether the linear code generated by $\left(\begin{array}{c}
            G_{\C_k(\SSS,{\bf v},\infty)} \\
            (\mathbf{v},v_{n+1})\star \mathbf{g}
        \end{array}\right)$ is MDS, we need to consider the following two cases.

        {\bf Case 1:} {\em $M_{k+1}$ contains the last column 
        of the matrix given in \eqref{eq.matrix3_deephole}.} 
        In this case, $M_{k+1}$ has the following form 
        \begin{align*}
           \left(\begin{array}{ccccc}
               v_{i_1}  & v_{i_2}  & \ldots & v_{i_{k}} & 0\\
               v_{i_1}a_{i_1}  &  v_{i_2}a_{i_2}  & \ldots & v_{i_k}a_{i_{k}} & 0\\
               \vdots  & \vdots  & \ddots & \vdots & \vdots\\
               v_{i_1}a^{k-2}_{i_1}  & v_{i_2}a^{k-2}_{i_2}  & \ldots & v_{i_k}a^{k-2}_{i_{k}} & 0\\
               v_{i_1}a^{k}_{i_1}  & v_{i_2}a^{k}_{i_2}  & \ldots & v_{i_k}a^{k}_{i_{k}} & 1\\
               v_{i_1} (g_{k+1}a^{k+1}_{i_1}+g_{k-1}a^{k-1}_{i_1}) & v_{i_2} (g_{k+1}a^{k+1}_{i_2}+g_{k-1}a^{k-1}_{i_2}) & \ldots & v_{i_k}(g_{k+1}a^{k+1}_{i_k}+g_{k-1}a^{k-1}_{i_k}) & u_{n+1}v_{n+1}-f_k\\
           \end{array}\right), 
        \end{align*}
        where $\SSS_{k}=\{a_{i_1}, \ldots, a_{i_{k}}\}\subseteq \SSS$. 
        It then follows from Lemma \ref{lem.det of generalized Vandermone matrix} that 
        the determinant of $M_{k+1}$ is
        \begin{align*}
             & -g_{k-1} \det \left(\begin{array}{ccccc}
                    v_{i_1}  & v_{i_2}  &  \ldots & v_{i_{k}} & 0 \\
                    v_{i_1}a_{i_1}  & v_{i_2}a_{i_2} &  \ldots & v_{i_{k}}a_{i_{k}} & 0 \\
                    \vdots  & \vdots & \ddots & \vdots & \vdots\\
                    v_{i_1}a^{k-1}_{i_1} & v_{i_2}a^{k-1}_{i_2} &  \ldots & v_{i_{k}}a^{k-1}_{i_{k}} & 0\\
                    v_{i_1}a^{k}_{i_1} & v_{i_2}a^{k}_{i_2} &  \ldots & v_{i_{k}}a^{k}_{i_{k}} & 1\\
                \end{array} \right)+ \det \left(\begin{array}{ccccc}
                    v_{i_1}  & v_{i_2}  &  \ldots & v_{i_{k}} & 0 \\
                    v_{i_1}a_{i_1}  & v_{i_2}a_{i_2} &  \ldots & v_{i_{k}}a_{i_{k}} & 0 \\
                    \vdots  & \vdots & \ddots & \vdots & \vdots\\
                    v_{i_1}a^{k-2}_{i_1} & v_{i_2}a^{k-2}_{i_2} &  \ldots & v_{i_{k}}a^{k-2}_{i_{k}} & 0\\
                    v_{i_1}a^{k}_{i_1} & v_{i_2}a^{k}_{i_2} &  \ldots & v_{i_{k}}a^{k}_{i_{k}} & 1\\
                    v_{i_1}g_{k+1} a^{k+1}_{i_1} & v_{i_2}g_{k+1} a^{k+1}_{i_2} &  \ldots & v_{i_{k}}g_{k+1} a^{k+1}_{i_{k}} & u_{n+1}v_{n+1}-f_k\\
                \end{array} \right) \\ 
                = & -\prod_{j=1}^{k}v_{i_j}\prod_{1\leq s<t\leq k}(a_{i_t}-a_{i_s})g_{k-1}+\\
                  & \quad \quad \prod_{j=1}^{k}v_{i_j}
                \left(
                    -g_{k+1} \det \left(\begin{array}{cccc}
                        1  & 1  & \ldots & 1  \\
                        a_{i_1}  & a_{i_2} & \ldots & a_{i_{k}}  \\
                        \vdots  & \vdots & \ddots  & \vdots\\
                        a^{k-2}_{i_1} & a^{k-2}_{i_2} & \ldots & a^{k-2}_{i_{k}} \\
                        a^{k+1}_{i_1} & a^{k+1}_{i_2} & \ldots & a^{k+1}_{i_{k}} \\
                    \end{array} \right)
                    +
                    (u_{n+1}v_{n+1}-f_k) \det\left(\begin{array}{cccc}
                        1  & 1  & \ldots & 1  \\
                        a_{i_1}  & a_{i_2} & \ldots & a_{i_{k}}  \\
                        \vdots  & \vdots & \ddots  & \vdots\\
                        a^{k-2}_{i_1} & a^{k-2}_{i_2} & \ldots & a^{k-2}_{i_{k}} \\
                        a^{k}_{i_1} & a^{k}_{i_2} & \ldots & a^{k}_{i_{k}} \\
                    \end{array} \right)
                \right)\\                
                = & \prod_{j=1}^{k} v_{i_j} \prod_{1\leq s<t\leq k}(a_{i_t}-a_{i_s})
                \left(
                    -g_{k-1}-
                    g_{k+1}\det\left(
                        \begin{array}{cc}
                            \sigma_1(\SSS_k) & \sigma_0(\SSS_k) \\
                            \sigma_2(\SSS_k) & \sigma_1(\SSS_k)
                        \end{array}
                    \right)+
                    (u_{n+1}v_{n+1}-f_k)\sigma_1(\SSS_k)
                \right)\\
                = & \prod_{j=1}^{k} v_{i_j} \prod_{1\leq s<t\leq k}(a_{i_t}-a_{i_s})
                \left(
                g_{k+1}\left(\sigma_2(\SSS_k)-\sigma_1(\SSS_k)^2\right)+
                (u_{n+1}v_{n+1}-f_k)\sigma_1(\SSS_k)-
                g_{k-1}
                \right).
         \end{align*}
         Since $\prod_{j=1}^{k} v_{i_j} \prod_{1\leq s<t\leq k}(a_{i_t}-a_{i_s})\neq 0$ 
         for any $\SSS_{k}=\{a_{i_1}, a_{i_2}, \ldots, a_{i_{k}}\}\subseteq \SSS$, we deduce that 
         $\det(M_{k+1})\neq 0$ if and only if 
         $$
         g_{k-1}\neq g_{k+1}\left(\sigma_2(\SSS_k)-\sigma_1(\SSS_k)^2\right)+(u_{n+1}v_{n+1}-f_k)\sigma_1(\SSS_k).
         $$

        {\bf Case 2:} {\em $M_{k+1}$ does not contain the last column 
        of the matrix given in \eqref{eq.matrix3_deephole}.} 
        Hence, by an argument similar to {\bf Case 1}, we get that  
            $$\det(M_{k+1})  =  \prod_{j=1}^{k+1} v_{i_j} \prod_{1\leq s<t\leq k+1}(a_{i_t}-a_{i_s})
                    \left( g_{k+1}\sigma_2(\SSS_{k+1})-g_{k-1} \right)$$ 
        and $\det(M_{k+1})\neq 0$ 
        if and only if $$g_{k-1}\neq g_{k+1}\sigma_2(\SSS_{k+1}).$$  
        
    Combining {\bf Cases 1} and {\bf 2}, it follows from Theorem \ref{th.deep hole} and Lemma \ref{lem.MDS} 
    that $(\mathbf{v},v_{n+1})\star \mathbf{g}$ is a deep hole of 
    the ESGRS code $\C_k(\SSS,{\bf v},\infty)$ 
    if and only if 
    \begin{align*}
        g_{k-1}\notin 
        \left\{g_{k+1}\left(\sigma_2(\SSS_{k})-\sigma_1(\SSS_{k})^2\right)+(u_{n+1}v_{n+1}-f_k)\sigma_1(\SSS_{k}):~\SSS_{k} \subseteq \SSS   \right\} 
        \cup \left\{g_{k+1}\sigma_2(\SSS_{k+1}):~ \SSS_{k+1} \subseteq \SSS   \right\}.
    \end{align*}
    This completes the proof. 
\end{IEEEproof}

\begin{remark}\label{rem.existence specific deep hole2}
    Let the notation be the same as that used in Theorem \ref{th.specific deep hole2}. 
    We have the following remarks for Theorem \ref{th.specific deep hole2}. 
    \begin{enumerate}
        \item 
        In Theorem \ref{th.specific deep hole2}, once 
        $g_{k+1}\in\F_q^*$, $u_{n+1},v_{n+1}\in\F_q$, and 
        $f(x)\in\mathcal{V}_k$ are fixed, the condition for 
        $(\mathbf{v},v_{n+1})\star \mathbf{g}$ to be a deep hole is equivalent 
        to choosing $g_{k-1}$ outside the following set 
        \begin{align*}
            \mathcal{L}
            =
            \left\{
            g_{k+1}\left(\sigma_2(\SSS_{k})-\sigma_1(\SSS_{k})^2\right)
            +(u_{n+1}v_{n+1}-f_k)\sigma_1(\SSS_{k}):~\SSS_{k} \subseteq \SSS
            \right\} 
            \cup
            \left\{
            g_{k+1}\sigma_2(\SSS_{k+1}):~\SSS_{k+1} \subseteq \SSS
            \right\}\subseteq \F_q.
        \end{align*}
        Hence, if $|\mathcal{L}|<q$, then there exists 
        $g_{k-1}\in\F_q\setminus\mathcal{L}$, and consequently 
        $(\mathbf{v},v_{n+1})\star \mathbf{g}$ is a deep hole of 
        $\C_k(\SSS,{\bf v},\infty)$. 
    A crude upper bound is $|\mathcal{L}|\leq \binom{n}{k}+\binom{n}{k+1}=\binom{n+1}{k+1}.$
    Hence, the numerical condition
    $
        q>\binom{n+1}{k+1} 
    $
    is sufficient for the existence of an admissible \(g_{k-1}\) in the second 
    class of deep holes. This condition can be quite restrictive when \(n\) is close to \(q\), 
    and should not be viewed as a necessary condition for obtaining explicit deep 
    holes. In fact, Theorem~\ref{th.specific deep hole1} already gives an explicit 
    class of deep holes without this restriction. 

        \item 
If one formally sets $g_{k+1}=0$ in the determinant computation of 
Theorem~\ref{th.specific deep hole2}, then the corresponding condition becomes
\[
g_{k-1}\notin 
\{0\}\cup
\left\{(u_{n+1}v_{n+1}-f_k)\sigma_1(\SSS_k):~\SSS_k\subseteq \SSS\right\}.
\]
Equivalently, this gives exactly the two alternatives in 
Theorem~\ref{th.specific deep hole1}. 
In other words, such a specialization recovers exactly the criterion in Theorem \ref{th.specific deep hole1}. 
We nevertheless state Theorem \ref{th.specific deep hole1} separately, 
since its proof also reveals the connection with EGRS codes and Roth-Lempel codes.    
\end{enumerate}
\end{remark}

\subsection{Three families of non-GRS MDS-NMDS codes with larger lengths and dimensions}

In this subsection, we apply deep holes obtained in subsection \ref{subsec.deep-holes} to derive 
more non-GRS MDS-NMDS codes, and explore their connections with Roth-Lempel codes.

Note that Li \emph{et al.} showed in Theorem~8 of \cite{LSZ2024} that an ESGRS code is not monomially equivalent to any GRS code,  
while the non-GRS property mentioned in this paper requires that the code is also not monomially equivalent to any EGRS code. 
In addition, for any $[n+1,k]_q$ EGRS code $\GRS_k(\SSS',{\bf v}',\infty)$ with $n\le q-1$, 
it is known that there exist $\overline{\SSS'}\subseteq \F_q$ and $\overline{{\bf v}'}\in(\F_q^*)^{n+1}$ such that
\begin{align}\label{eq.GRS-EGRS}
   \GRS_k(\SSS',{\bf v}',\infty)=\GRS_k(\overline{\SSS'},\overline{{\bf v}'}). 
\end{align}
Since GRS codes have lengths at most $q$, to prove that $\C_k(\SSS,{\bf v},\infty)$ is non-GRS in the context of this paper, 
it suffices to show that $\C_k(\F_q,{\bf v},\infty)$ is not monomially equivalent to any $[q+1,k,q-k+2]_q$ EGRS code.
This motivates the following result. 

\begin{lemma}\label{lem.q+1}
    A $[q+1,k]_q$ ESGRS code $\C_k(\F_q,{\bf v},\infty)$ is not monomially equivalent to any EGRS code. 
    Specifically, we have the following results. 
    \begin{enumerate}
        \item $\C_k(\F_q,{\bf v},\infty)$ is a $[q+1,k,q-k+1]_q$ NMDS code for any $(p,k)\neq (2,q-2)$, 
        where $p$ is the characteristic of $\F_q$.   
        \item $\C_{q-2}(\F_q,{\bf v},\infty)$ is a $[q+1,q-2,4]_q$ MDS code that is not monomially equivalent to any EGRS code 
        for even $q\geq 8$.   
    \end{enumerate}
\end{lemma}
\begin{IEEEproof}
Recall that the condition $3\le k\le q-1$ is required for the code $\C_k(\F_q,{\bf v},\infty)$ to be well defined. 
Under this condition, it follows from \cite[Corollary~1]{HF2023} that, except for the special case 
$k=q-2$ with even $q\geq 8$, the finite field $\F_q$ contains a $k$-zero-sum subset. 
Hence, by Lemma~\ref{lem.must be MDS or NMDS}, the parameters of $\C_k(\F_q,{\bf v},\infty)$ presented in 1) and 2) hold. 

Given any $[n_1,k_1]_q$ linear code $\C$, we let 
$\C^2=\{(c_1^2,c_2^2,\ldots,c_{n_1}^2):~(c_1,c_2,\ldots,c_{n_1})\in \C\}$ 
be the square code of $\C$. 
Note that it has been proved in \cite[Theorem 7.2)]{LSZ2024} that $d((\C_{q-2}(\F_q,{\bf 1},\infty)^{\perp})^2)=1$ and 
it is well-known that $d((\EGRS_{q-2}(\F_q,{\bf v},\infty)^{\perp})^2)\geq 2$ for any ${\bf v}\in (\F_q^*)^{q}$ and $q\geq 5$ 
(see also \cite[Proposition 2.1(2)]{ZL2024}), 
where ${\bf 1}$ is the all-one vector of length $q$.  
Since monomially equivalent linear codes have monomially equivalent square codes, 
and the dual of an ESGRS code is monomially equivalent to another ESGRS code, 
it follows that $\C_{q-2}(\F_q,{\bf v},\infty)$ is not monomially equivalent to any EGRS code, as their square codes have different minimum distances.
This completes the proof.
\end{IEEEproof}

\begin{remark}
    Since duals of non-GRS MDS codes are also non-GRS MDS codes, 
    Lemma~\ref{lem.q+1} implies the existence of $[q+1,k,q-k+2]_{q}$ 
    non-GRS MDS codes for $k\in\{3,q-2\}$ and every even $q\geq 8$. 
    In particular, the $[q+1,3,q-1]_q$ non-GRS MDS code is not from the family of ESGRS codes. 
    The length $q+1$ in these codes is larger than the lengths obtained from several known 
    constructions under comparable settings, such as the length $\frac{q}{2}$ from twisted GRS codes 
    in~\cite{BPR2022}, the length $O(q^{1/k})$ from algebraic geometry codes with higher genus curves 
    under certain conditions in~\cite{C2024} and from polynomial evaluation codes with polynomial spaces 
    spanned by monomials in~\cite{JMXZ2025}, and the length $\frac{q}{2}+\sqrt{q}$ from elliptic curve codes 
    in~\cite{HR2023}. 

\end{remark}

\begin{theorem}{\rm (An improvement of \cite[Theorem \Rmnum{3}.7]{LSZ2024})}\label{th.non-GRS}
    Any $[n+1,k]_q$ ESGRS code $\C_k(\SSS,{\bf v},\infty)$ is non-GRS.  
\end{theorem}
\begin{IEEEproof}
We consider the following two cases. 
\begin{itemize}
    \item If $n\le q-1$, then the length of 
$\C_k(\SSS,\vvv,\infty)$ is $n+1\le q$. 
As we discussed above, $\C_k(\SSS,\vvv,\infty)$ is not monomially equivalent to any GRS code by \cite[Theorem~8]{LSZ2024}. 
With \eqref{eq.GRS-EGRS}, we also know that every EGRS code of length at most $q$ is monomially equivalent to a GRS code. 
It then implies that $\C_k(\SSS,\vvv,\infty)$ is not monomially equivalent to any GRS or EGRS code.

\item If $n=q$, then $\C_k(\SSS,\vvv,\infty)=\C_k(\mathbb F_q,\vvv,\infty)$. 
Since every $q$-ary GRS code has length at most $q$, while 
$\C_k(\F_q,\vvv,\infty)$ has length $q+1$, it cannot be monomially 
equivalent to any GRS code.
Then combining with Lemma \ref{lem.q+1} directly yields the desired result.
\end{itemize}
This completes the proof. 
\end{IEEEproof}

Based on Theorems \ref{th.connection deep holes and extened codes} and \ref{th.non-GRS}, 
we now can apply Theorems \ref{th.specific deep hole1} and \ref{th.specific deep hole2} to derive 
more families of non-GRS MDS-NMDS codes. 
First, under the condition described in Theorem \ref{th.specific deep hole1}.1), 
we obtain a class of $[n+2,k+1]_q$ non-GRS MDS-NMDS codes,  
which are monomially equivalent to certain Roth-Lempel codes. 
In other words, some non-GRS MDS-NMDS Roth-Lempel codes can be reconstructed via deep holes of ESGRS codes as follows.

\begin{theorem}\label{th.more non-GRS MDS codes111}
    Suppose that $\mathbf{g}=(g(a_1),g(a_2),\ldots,g(a_{n}),u_{n+1})\in \F_q^{n+1}$, where  
    $g(x)=g_{k-1}x^{k-1}+f(x)\in \F_q[x]$ with $g_{k-1}\in \F_q^*$ and $f(x)\in \mathcal{V}_k$. 
    Let $\C'$ be the linear code generated by $\left(\begin{array}{cc}
        G_{\C_k(\SSS,{\bf v},\infty)} & {\bf 0}\\
        (\mathbf{v},v_{n+1})\star \mathbf{g} & 1
    \end{array}
     \right)$, 
     where $G_{\C_k(\SSS,{\bf v},\infty)}$ is a generator matrix of the ESGRS code $\C_k(\SSS,{\bf v},\infty)$ 
     given in \eqref{eq.ESGRS_generator_matrix}. 
     If $u_{n+1}v_{n+1}=f_k$ with $v_{n+1}\in \F_q$, then $\C'$ is monomially equivalent to 
     the Roth-Lempel code $\RL_{k+1,0}(\SSS)$.  Moreover, we have the following results. 
     \begin{enumerate}
        \item If $\SSS$ is $k$-zero-sum free, then $\C'$ is an $[n+2,k+1,n-k+2]_q$ non-GRS MDS code. 
        \item If $\SSS$ contains a $k$-zero-sum subset, then $\C'$ is an $[n+2,k+1,n-k+1]_q$ NMDS code.
     \end{enumerate}
     In summary, we obtain another family of non-GRS MDS-NMDS codes.
\end{theorem}
\begin{IEEEproof}
    From Lemma \ref{lem.covering radius}, the ESGRS code $\C_k(\SSS,{\bf v},\infty)$ 
    has covering radius $\rho(\C_k(\SSS,{\bf v},\infty))=n-k+1$. 
    It then follows from Lemma \ref{lem.must be MDS or NMDS}, 
    and Theorems \ref{th.connection deep holes and extened codes},  
    \ref{th.specific deep hole1}.1) and \ref{th.non-GRS} that 
    $\C'$ is an $[n+2,k+1,n-k+2]_q$ non-GRS MDS code or an $[n+2,k+1,n-k+1]_q$ NMDS code 
    under the given conditions in result 1) or 2), respectively. 
    Since $u_{n+1}v_{n+1}=f_k$, we get the following matrix relationships 
\begin{align*}
    \left(\begin{array}{cc}
        G_{\C_k(\SSS,{\bf v},\infty)} & {\bf 0}\\
        (\mathbf{v},v_{n+1})\star \mathbf{g} & 1
    \end{array}
     \right) & \overset{r}{\backsimeq} 
             \left(\begin{array}{cccccc}
            v_1 & v_2 &  \ldots & v_{n} & 0 & 0\\ 
            v_1a_1  & v_2a_2  & \ldots & v_{n}a_{n} & 0 & 0\\
            \vdots & \vdots &  \ddots & \vdots & \vdots & \vdots\\ 
            v_1a_1^{k-2}  & v_2a_2^{k-2}  & \ldots & v_{n}a_{n}^{k-2} & 0 & 0\\ 
            v_1a_1^{k-1}  & v_2a_2^{k-1}  & \ldots & v_{n}a_{n}^{k-1} & g_{k-1}^{-1}(u_{n+1}v_{n+1}-f_k) & g_{k-1}^{-1} \\ 
            v_1a_1^{k}  & v_2a_2^{k}  & \ldots & v_{n}a_{n}^{k} & 1 & 0
        \end{array}\right) \\ & \overset{c}{\backsimeq} 
        \left(\begin{array}{cccccc}
            1 & 1 &  \ldots & 1 & 0 & 0\\ 
            a_1  & a_2  & \ldots & a_{n} & 0 & 0\\
            \vdots & \vdots &  \ddots & \vdots & \vdots & \vdots\\ 
            a_1^{k-2}  & a_2^{k-2}  & \ldots & a_{n}^{k-2} & 0 & 0\\ 
            a_1^{k-1}  & a_2^{k-1}  & \ldots & a_{n}^{k-1} & 0 & 1 \\ 
            a_1^{k}  & a_2^{k}  & \ldots & a_{n}^{k} & 1 & 0
        \end{array}\right),
    \end{align*}
    where ``$\overset{r}{\backsimeq}$" and ``$\overset{c}{\backsimeq}$" denote the row and column equivalences, respectively. 
    Note that the last matrix is exactly the generator matrix of the Roth-Lempel code $\RL_{k+1,0}(\SSS)$ given in 
    \eqref{eq.ESGRS_generator_matrix of Roth-Lempel code},
    which implies that $\C'$ is monomially equivalent to $\RL_{k+1,0}(\SSS)$. 
    This completes the proof. 
\end{IEEEproof}

The following theorem yields another two classes of $[n+2,k+1]_q$ 
non-GRS MDS-NMDS codes with respect to the conditions described in Theorems \ref{th.specific deep hole1}.2) 
and \ref{th.specific deep hole2}.

\begin{theorem}\label{th.more non-GRS MDS codes222}
    Suppose that $\mathbf{g}=(g(a_1),g(a_2),\ldots,g(a_{n}),u_{n+1})\in \F_q^{n+1}$, where  
    $g(x)=g_{k+1}x^{k+1}+g_{k-1}x^{k-1}+f(x)\in \F_q[x]$ with not both $g_{k+1}$ and $g_{k-1}$ equal to $0$, 
    and $f(x)\in \mathcal{V}_k$. 
    Let $\C'$ be the linear code generated by $\left(\begin{array}{cc}
        G_{\C_k(\SSS,{\bf v},\infty)} & {\bf 0}\\
        (\mathbf{v},v_{n+1})\star \mathbf{g} & 1
    \end{array}
     \right)$, 
     where $G_{\C_k(\SSS,{\bf v},\infty)}$ is a generator matrix of the ESGRS code $\C_k(\SSS,{\bf v},\infty)$ 
     given in \eqref{eq.ESGRS_generator_matrix}. 
     If one of the following conditions holds: 
     \begin{itemize}
        \item [\rmnum{1})] $g_{k+1}=0$, $g_{k-1}\neq 0$, 
        $\SSS$ is an  $(n,k,g_{k-1}(u_{n+1}v_{n+1}-f_k)^{-1})$-set, and $u_{n+1}v_{n+1}\neq f_k$ with $v_{n+1}\in \F_q$;
        \item [\rmnum{2})] $g_{k+1}\neq 0$, and 
        $$g_{k-1}\notin 
        \left\{g_{k+1}\left(\sigma_2(\SSS_{k})-\sigma_1(\SSS_{k})^2\right)+(u_{n+1}v_{n+1}-f_k)\sigma_1(\SSS_{k}):~\SSS_{k} \subseteq \SSS   \right\} 
        \cup \left\{g_{k+1}\sigma_2(\SSS_{k+1}):~ \SSS_{k+1} \subseteq \SSS   \right\}$$ 
        with $v_{n+1}\in \F_q$,  
     \end{itemize}
     then the following statements hold.
    \begin{enumerate}
            \item [1)] If $\SSS$ is $k$-zero-sum free, then $\C'$ is an $[n+2,k+1,n-k+2]_q$ non-GRS MDS code. 
            \item [2)] If $\SSS$ contains a $k$-zero-sum subset, then $\C'$ is an $[n+2,k+1,n-k+1]_q$ NMDS code.
    \end{enumerate}
    In summary, we obtain another two families of non-GRS MDS-NMDS codes.
\end{theorem}
\begin{IEEEproof}
    By an argument similar to that used in the proof of Theorem \ref{th.more non-GRS MDS codes111}, 
    we immediately obtain that $\C'$ is either an $[n+2,k+1,n-k+2]_q$ non-GRS MDS code 
    or an $[n+2,k+1,n-k+1]_q$ NMDS code by combining Theorems \ref{th.specific deep hole1}.2) 
    and \ref{th.specific deep hole2}. 
    This completes the proof.  
\end{IEEEproof}

\begin{remark}\label{rem.answer_to_Problem1}
    We have two remarks for Theorems \ref{th.more non-GRS MDS codes111} 
and \ref{th.more non-GRS MDS codes222} as follows. 
\begin{enumerate}
    \item Note that Theorems~\ref{th.more non-GRS MDS codes111} and~\ref{th.more non-GRS MDS codes222} 
yield three families of $[n+2,k+1]_q$ non-GRS MDS-NMDS codes.
In particular, we prove that the code family constructed in Theorem~\ref{th.more non-GRS MDS codes111} 
is monomially equivalent to the Roth-Lempel codes $\RL_{k+1,0}(\SSS)$.
Consequently, it follows immediately that Roth-Lempel codes contain a distinguished subclass of non-GRS MDS-NMDS codes.

\item  Unlike for Theorem \ref{th.more non-GRS MDS codes111}, 
        it is not straightforward to determine the monomial equivalence of 
        non-GRS MDS codes $\C'$ derived from Theorem \ref{th.more non-GRS MDS codes222} 
        and Roth-Lempel codes. 
        Taking Theorem \ref{th.more non-GRS MDS codes222}.1) as an example, 
        it is easily seen that the $[n+2,k+1,n-k+2]_q$ non-GRS MDS code $\C'$ obtained  
        is further monomially equivalent to the linear code generated by the matrix 
            \begin{align*}
                G_{\C'}=\left(\begin{array}{cccccc}
                    1 & 1 &  \ldots & 1 & 0 & 0\\ 
                    a_1  & a_2 & \ldots & a_{n} & 0 & 0\\
                    \vdots & \vdots &  \ddots & \vdots & \vdots & \vdots\\ 
                    a_1^{k-2}  & a_2^{k-2} & \ldots & a_{n}^{k-2} & 0 & 0\\ 
                    a_1^{k-1}  & a_2^{k-1} & \ldots & a_{n}^{k-1} & 1 & 1 \\ 
                    a_1^{k}  & a_2^{k} & \ldots & a_{n}^{k} & 0 & g_{k-1}(u_{n+1}v_{n+1}-f_k)^{-1}
                \end{array}\right).
            \end{align*}
            Consider $\SSS=\{1, 2, 5, 6, 9\}\subseteq \F_{11}$, 
            $\vvv=(1,1,1,1,1)\in (\F_{11}^*)^5$, and 
            $g_{k-1}(u_{n+1}v_{n+1}-f_k)^{-1}=3$. 
            Verified by Magma \cite{Magma}, we have the following facts: 
            \begin{itemize}
                \item $\C'$ is a $[7,4,4]_{11}$ non-GRS MDS code that 
                is monomially equivalent to the Roth-Lempel code $\RL_{4,\delta}(\SSS)$ for $\delta\in \{3,7\}$; 
    
                \item $\C'$ is a $[7,4,4]_{11}$ non-GRS MDS code that 
                is not monomially equivalent to any Roth-Lempel code 
                $\RL_{4,\delta}(\SSS)$ for $\delta\in \{0,1,2,4,5,6,8,9,10\}$.
            \end{itemize}     
\end{enumerate}
\end{remark}

In the following, we provide an algorithm for obtaining more non-GRS MDS-NMDS codes based on 
Theorems \ref{th.more non-GRS MDS codes111} and \ref{th.more non-GRS MDS codes222}. 
The procedures for producing non-GRS MDS codes and NMDS codes are identical, except that
the input distance parameter is $d=n-k+2$ or $d=n-k+1$, respectively.
For conciseness, we present a unified algorithm in Algorithm~\ref{alg.1}.

\begin{theorem}\label{th.alg2}
    More non-GRS MDS $[n+2,k+1,n-k+2]_q$ codes and $[n+2,k+1,n-k+1]_q$ NMDS codes can be obtained 
    by Algorithm \ref{alg.1} in a unified way, respectively.
\end{theorem}
\begin{IEEEproof}
    The desired result follows immediately from Theorems 
    \ref{th.more non-GRS MDS codes111} and \ref{th.more non-GRS MDS codes222}. 
    More details can be found in Algorithm \ref{alg.1}.  
\end{IEEEproof}

\begin{algorithm*}[t]
\caption{A unified algorithm for obtaining more non-GRS MDS-NMDS codes from Theorems~\ref{th.more non-GRS MDS codes111} and~\ref{th.more non-GRS MDS codes222}}
\label{alg.1}
\KwIn{An $[n+1,k,d]_q$ ESGRS code $\C_k(\SSS,\vvv,\infty)$, where $d\in\{n-k+2,\;n-k+1\}$}
\KwOut{An $[n+2,k+1,d]_q$ linear code $\C'$, which is either a non-GRS MDS code (if $d=n-k+2$) or an NMDS code (if $d=n-k+1$)}

\BlankLine
\For{$(g_{k+1}, g_{k-1}, f(x), u_{n+1}, v_{n+1})
      \in \F_q \times \F_q \times \mathcal{V}_k \times \F_q \times \F_q$}
{
$g(x)\gets g_{k+1}x^{k+1}+g_{k-1}x^{k-1}+f(x)$\;
${\bf g}\gets \bigl(g(a_1),g(a_2),\ldots,g(a_n),u_{n+1}\bigr)$\;
$f_k\gets$ the coefficient of $x^k$ in $f(x)$\;

$\C' \gets$ the linear code generated by
$\left(\begin{array}{cc}
        G_{\C_k(\SSS,{\bf v},\infty)} & {\bf 0}\\
        (\mathbf{v},v_{n+1})\star \mathbf{g} & 1
    \end{array}
     \right)$\;

\uIf{$g_{k+1}=0$, $g_{k-1}\neq 0$ and $u_{n+1}v_{n+1}=f_k$}{
\Return $\C'$ is an $[n+2,k+1,d]_q$ code that is monomially equivalent to the Roth-Lempel code $\RL_{k+1,0}(\SSS)$\;
}
\uElseIf{$g_{k+1}=0$, $g_{k-1}\neq 0$, $u_{n+1}v_{n+1}\neq f_k$, and $\SSS$ is an $(n,k,g_{k-1}(u_{n+1}v_{n+1}-f_k)^{-1})$-set}{
\Return $\C'$ is an $[n+2,k+1,d]_q$ code\;
}
\ElseIf{$g_{k+1}\neq 0$}{
$L_1\gets \emptyset$,\quad $L_2\gets \emptyset$\;

\ForEach{$\SSS_k\subseteq \SSS$ with $|\SSS_k|=k$}{
$L_1 \gets L_1 \cup \Bigl\{\, g_{k+1}\bigl(\sigma_2(\SSS_k)-\sigma_1(\SSS_k)^2\bigr)
+(u_{n+1}v_{n+1}-f_k)\sigma_1(\SSS_k)\,\Bigr\}$\;
}
\ForEach{$\SSS_{k+1}\subseteq \SSS$ with $|\SSS_{k+1}|=k+1$}{
$L_2 \gets L_2 \cup \Bigl\{\, g_{k+1}\sigma_2(\SSS_{k+1})\,\Bigr\}$\;
}
\If{$g_{k-1}\notin L_1\cup L_2$}{
\Return $\C'$ is an $[n+2,k+1,d]_q$ code\;
}
}

}
\end{algorithm*}

We end this section with the following two examples.

\begin{example}\label{exam.4}
    Let $q=11$, $\SSS=\{3, 4, 5, 6, 7\}\subseteq \F_{11}$, 
    and $\vvv=(1,1,1,1,1)\in (\F_{11}^*)^{5}$. 
    Note that $\SSS$ is $3$-zero-sum free. 
    It then follows from Lemma \ref{lem.must be MDS or NMDS} and Theorem \ref{th.non-GRS} 
    that the ESGRS code $\C_{3}(\SSS,\vvv,\infty)$ is 
    a $[6,3,4]_{11}$ non-GRS MDS code with generator matrix 
    $$
    G_3=\left(
        \begin{array}{ccccccc}
            1 & 1 & 1 & 1 & 1 & 0\\
            3 & 4 & 5 & 6 & 7 & 0\\
            5 & 9 & 4 & 7 & 2 & 1
        \end{array}
    \right).
    $$
    From Theorem \ref{lem.covering radius}, we have $$\rho(\C_{3}(\SSS,\vvv,\infty))=3.$$
     Assume that $g_4=0$, $g_2=3$, and $f(x)=4x^3+10x+7\in \mathcal{V}_3$. 
    Then $g(x)=4x^3+3x^2+10x+7$. 
    Furthermore, we have the following results. 
    \begin{enumerate}
        \item  According to Theorem \ref{th.specific deep hole1}.1), we deduce that 
        \begin{align*}
             (g(a_1),g(a_2),g(a_3),g(a_4),g(a_5),4) 
          =  (7,10,5,5,1,4)
        \end{align*}
        is a deep hole of $\C_{3}(\SSS,\vvv,\infty)$. 

        \item Note also that $\SSS$ is a $(5,3,\delta)$-set for $\delta\in \{0,8,9,10\}$. 
        It then follows from Theorem \ref{th.specific deep hole1}.2) that 
         \begin{align*}
             (g(a_1),g(a_2),g(a_3),g(a_4),g(a_5),u_{6}v_{6}) 
            = \left\{
                \begin{array}{ll}
                    (7,10,5,5,1,3), & {\rm if}~ 3(u_{6}v_{6}-4)^{-1}=8, \\ 
                    (7,10,5,5,1,8), & {\rm if}~ 3(u_{6}v_{6}-4)^{-1}=9, \\ 
                    (7,10,5,5,1,1), & {\rm if}~ 3(u_{6}v_{6}-4)^{-1}=10
                \end{array}
            \right.
        \end{align*}
       are three deep holes of $\C_{3}(\SSS,\vvv,\infty)$. 
       In addition, $\C_{3}(\SSS,\vvv,\infty)$ does not have 
       other deep holes related to the polynomial $g(x)$. 
        
        \item From Theorem \ref{th.more non-GRS MDS codes111}, 
        we directly obtain a $[7,4,4]_{11}$ non-GRS MDS code generated by the matrix 
        \begin{align*}
            \left(
                \begin{array}{cccccccc}
                    1 & 1 & 1 & 1 & 1 & 0 & 0\\
                    3 & 4 & 5 & 6 & 7 & 0 & 0\\
                    5 & 9 & 4 & 7 & 2 & 1 & 0\\
                    7 & 10 & 5 & 5 & 1 & 4 & 1
                \end{array}
            \right),
        \end{align*}
        which is monomially equivalent to the Roth-Lempel code $\RL_{4,0}(\SSS)$.

        \item From Theorem \ref{th.more non-GRS MDS codes222}.1), 
        we directly obtain three $[7,4,4]_{11}$ non-GRS MDS codes $\C'_1$, $\C'_2$, and $\C'_3$ 
        generated by the matrices 
        \begin{align*}
            \left(
                \begin{array}{cccccccc}
                    1 & 1 & 1 & 1 & 1 & 0 & 0\\
                    3 & 4 & 5 & 6 & 7 & 0 & 0\\
                    5 & 9 & 4 & 7 & 2 & 1 & 0\\
                    7 & 10 & 5 & 5 & 1 & 1 & 1
                \end{array}
            \right),~
            \left(
                \begin{array}{cccccccc}
                    1 & 1 & 1 & 1 & 1 & 0 & 0\\
                    3 & 4 & 5 & 6 & 7 & 0 & 0\\
                    5 & 9 & 4 & 7 & 2 & 1 & 0\\
                    7 & 10 & 5 & 5 & 1 & 3 & 1
                \end{array}
            \right),~{\rm and}~
            \left(
            \begin{array}{cccccccc}
                1 & 1 & 1 & 1 & 1 & 0 & 0\\
                3 & 4 & 5 & 6 & 7 & 0 & 0\\
                5 & 9 & 4 & 7 & 2 & 1 & 0\\
                7 & 10 & 5 & 5 & 1 & 8 & 1
            \end{array}
            \right),
        \end{align*}
        respectively. 
        Moreover, verified by Magma \cite{Magma}, we have the following facts:
        \begin{itemize}
            \item $\C'_1$ and $\C'_2$ are not monomially equivalent to 
            any Roth-Lempel code $\RL_{4,\delta}(\SSS)$ for any $\delta\in \F_{11}$;

            \item $\C'_3$ is monomially equivalent to the Roth-Lempel code 
            $\RL_{4,\delta}(\SSS)$ for $\delta\in \{9, 10\}$. 
        \end{itemize}
    \end{enumerate}
\end{example}

\begin{example}\label{exam.5}
    Let $\C_{3}(\SSS,\vvv,\infty)$ be the same $[6,3,4]_{11}$ non-GRS MDS ESGRS code as in Example \ref{exam.4}. 
    Assume that $g_{4}=2$ and $f(x)=3x^3+5x+2$. 
    Furthermore, we have the following results. 
    \begin{enumerate}
        \item     It can be checked that 
        \begin{align*}
            \begin{split}
                \left\{g_{4}\left(\sigma_2(\SSS_{3})-\sigma_1(\SSS_{3})^2\right)+(u_{6}v_{6}-f_3)\sigma_1(\SSS_{3}):~\SSS_{3} \subseteq \SSS   \right\} 
                \cup \left\{g_{4}\sigma_2(\SSS_{4}):~ \SSS_{4} \subseteq \SSS   \right\} 
            =  \left\{
                \begin{array}{ll}
                    \F_{11}\setminus \{8\}, & {\rm if}~ u_{6}v_{6}=0, \\
                    \F_{11}, & {\rm if}~ u_{6}v_{6}=4. 
                \end{array}
            \right.
            \end{split}
        \end{align*}
        The remaining cases $u_6v_6\in\F_{11}\setminus\{0,4\}$ can be computed in the same way. 
        For this example, we only use the two cases $u_6v_6=0$ and $u_6v_6=4$.           
        Hence, according to Theorem \ref{th.specific deep hole2}, for $u_{6}v_{6}=0$ and $g_2=8$, 
        $i.e.,$ $g(x)=2x^4+3x^3+8x^2+5x+2$, 
        \begin{align*}
            (g(a_1),g(a_2),g(a_3),g(a_4),g(a_5),u_{6}v_{6}) 
         =  (2,7,4,7,1,0)
       \end{align*}
       is a deep hole of $\C_{3}(\SSS,\vvv,\infty)$ 
       and $\C_{3}(\SSS,\vvv,\infty)$ does not have other deep holes associated with 
       $g(x)=2x^4+3x^3+g_2x^2+f_1x+f_0$ for any $g_2,f_0,f_1\in \F_{11}$ when $u_6v_6=4$.

    \item From Theorem \ref{th.more non-GRS MDS codes222}.2), 
    we directly obtain a $[7,4,4]_{11}$ non-GRS MDS code generated by the matrix 
    \begin{align*}
        \left(
            \begin{array}{cccccccc}
                1 & 1 & 1 & 1 & 1 & 0 & 0\\
                3 & 4 & 5 & 6 & 7 & 0 & 0\\
                5 & 9 & 4 & 7 & 2 & 1 & 0\\
                2 & 7 & 4 & 7 & 1 & 0 & 1
            \end{array}
        \right).
    \end{align*}
    Moreover, verified by Magma \cite{Magma}, 
    this non-GRS MDS code is not monomially equivalent to 
    any Roth-Lempel code $\RL_{4,\delta}(\SSS)$ for any $\delta\in \F_{11}$. 
    \end{enumerate}
\end{example}

\section{Conclusion}\label{sec.concluding remarks}

This paper addresses Question~\ref{ques.deep hole} by developing a unified framework 
that constructs new families of non-GRS MDS-NMDS codes from known ones via deep holes. 
We showed that more $[n+1,k+1]_q$ non-GRS MDS-NMDS codes 
can be directly obtained from known $[n,k]_q$ non-GRS MDS-NMDS codes with covering radius $n-k$ 
in Theorem \ref{th.connection deep holes and extened codes} and Corollary \ref{coro.family}, 
thereby enabling a simultaneous increase in both lengths and dimensions.  
We further reformulate this framework in terms of the second kind of extended codes 
\cite{SDC2024DM,SDC2024FFA} in Theorem~\ref{th.equivalent_framework}. 
{This formulation recovers a main result in~\cite{WDC2025} as a special case 
and provides a provable reduction in computational complexity compared with~\cite{MKZ2026}. 
It also reveals additional structural properties of the resulting codes.}  
As an application of this framework, we respectively determined the covering radius and 
two explicit classes of deep holes of ESGRS codes in Lemma \ref{lem.covering radius}, 
and Theorems \ref{th.specific deep hole1} and \ref{th.specific deep hole2}. 
With an improvement of the non-GRS property of ESGRS codes in Theorem \ref{th.non-GRS}, 
we obtained three more families of non-GRS MDS-NMDS codes in Theorems \ref{th.more non-GRS MDS codes111} 
and \ref{th.more non-GRS MDS codes222}. For easy use, we also provided an algorithm for generating these 
non-GRS MDS-NMDS codes in Theorem \ref{th.alg2} and Algorithm \ref{alg.1}. 
Finally, we discussed the monomial equivalence of these new non-GRS MDS-NMDS codes with Roth-Lempel codes and 
gave several examples to support our findings. 
It would also be interesting to determine other deep holes of ESGRS codes or explore 
new families of non-GRS MDS-NMDS codes with a described covering radius and explicit deep holes 
to derive more non-GRS MDS-NMDS codes via our framework.

\vfill

\end{document}